\documentclass[12pt,draftcls,onecolumn]{IEEEtran}
\usepackage{amssymb,bbm,amsmath,color,psfrag,subfigure}
\usepackage[draft]{graphicx}
\usepackage{xspace}
\usepackage{algorithm,algorithmic} 
\usepackage{wrapfig}
\usepackage{tikz}

\def\qed{\hfill \vrule height 7pt width 7pt depth 0pt\medskip}
\def\beq{\begin{equation}}
\def\eeq{\end{equation}}
\def\proof{\noindent{\bf Proof}\ \ }

\newcommand{\ds}{\displaystyle}

\newcommand{\ba}{\begin{array}}
\newcommand{\ea}{\end{array}}

\renewcommand{\l}{\left}\renewcommand{\r}{\right}
\newcommand{\be}{\begin{equation}}
\newcommand{\ee}{\end{equation}}
\newcommand{\eps}{\varepsilon}

\newcommand{\ov}{\overline}

\newcommand{\R}{\mathbb{R}}

\newcommand{\de}{\mathrm{d}}

\newcommand{\se}{\text{ if }}

\DeclareMathOperator{\sgn}{sgn}

\DeclareMathOperator{\cl}{cl}

\DeclareMathOperator*{\argmin}{argmin}

\newcommand{\real}{{\mathbb{R}}}
\newcommand{\realnoneg}{\R_+}

\newcommand{\subscr}[2]{{#1}_{\textup{#2}}}
\newcommand{\supscr}[2]{{#1}^{\textup{#2}}}

\newcommand{\mc}{\mathcal}

\newcommand{\flowmax}{\supscr{f}{max}}
\newcommand{\flowmaxe}{\subscr{f}{e}^\textup{max}}
\newcommand{\newflowmaxe}{\subscr{\tilde{f}}{e}^\textup{max}}

\newcommand{\flow}{f}
\newcommand{\floweq}{{\flow}^{\text{*}}}

\newcommand{\rhomedian}{{\rho}^{\mu}}
\newcommand{\tilderhomedian}{\tilde{\rho}^{\mu}}

\newcommand{\graph}{\mc T}

\newcommand{\onebf}{\mathbf{1}}

\newcommand{\ilogitconst}{\eta}

\newtheorem{assumption}{Assumption}

\newtheorem{theorem}{Theorem}
\newtheorem{proposition}{Proposition}

\newtheorem{definition}{Definition}
\newtheorem{lemma}{Lemma}
\newtheorem{remark}{Remark}

\newtheorem{example}{Example}

\newcommand\oprocendsymbol{\hbox{$\square$}}
\newcommand\oprocend{\relax\ifmmode\else\unskip\hfill\fi\oprocendsymbol}

\begin{document}

\title{Robust Distributed Routing in Dynamical Flow Networks --
Part I: Locally Responsive Policies and Weak Resilience}

\author{Giacomo Como\thanks{G.~Como, K.~Savla, M.A.~Dahleh and E.~Frazzoli are with the Laboratory for Information and Decision Systems at the Massachusetts Institute of Technology. \texttt{\{giacomo,ksavla,dahleh,frazzoli\}@mit.edu.}} \quad Ketan Savla \quad Daron Acemoglu\thanks{D.~Acemoglu is with the Department of Economics at the Massachusetts Institute of Technology. \texttt{daron@mit.edu}.} \quad Munther A. Dahleh \quad Emilio Frazzoli \thanks{This work was supported in part by NSF EFRI-ARES grant number 0735956. Any opinions, findings, and conclusions or recommendations expressed in this publication are those of the authors and do not necessarily reflect the views of the supporting organizations. G. Como and K. Savla thank Prof. Devavrat Shah for helpful discussions.
A preliminary version of this paper appeared in part as \cite{Como.Savla.ea:MTNS10}.
}} 

\maketitle

\begin{abstract}
Robustness of distributed routing policies is studied for dynamical flow networks, with respect to adversarial disturbances that reduce the link flow capacities. A dynamical flow network is modeled as a system of ordinary differential equations derived from mass conservation laws on a directed acyclic graph with a single origin-destination pair and a constant inflow at the origin. Routing policies regulate the way the inflow at a non-destination node gets split among its outgoing links as a function of the current particle density, while the outflow of a link is modeled to depend on the current particle density on that link through a flow function. The dynamical flow network is called partially transferring if the total inflow at the destination node is asymptotically bounded away from zero, and its weak resilience is measured as the minimum sum of the link-wise magnitude of all disturbances that make it not partially transferring. The weak resilience of a dynamical flow network with arbitrary routing policy is shown to be upper-bounded by the network's min-cut capacity, independently of the initial flow conditions. Moreover, a class of distributed routing policies that rely exclusively on local information on the particle densities, and are locally responsive to that, is shown to yield such maximal weak resilience.   These results imply that locality constraints on the information available to the routing policies do not cause loss of weak resilience. Some fundamental properties of dynamical flow networks driven by locally responsive distributed policies are analyzed in detail, including global convergence to a unique limit flow. 
\end{abstract}


\section{Introduction}
Flow networks provide a fruitful modeling framework for many applications of interest such as transportation, data, and production networks. They entail a fluid-like description of the macroscopic motion of \emph{particles}, which are routed from their origins to their destinations via intermediate nodes: we refer to standard textbooks, such as \cite{Ahuja.Magnanti.ea:93}, for a thorough treatment. 

The present and a companion paper \cite{PartII} study \emph{dynamical flow networks}, modeled as systems of ordinary differential equations derived from mass conservation laws on directed acyclic graphs with a single origin-destination pair and a constant inflow at the origin. The rate of change of the particle density on each link of the network equals the difference between the \emph{inflow} and the \emph{outflow} of that link. The latter is modeled to depend on the current particle density on that link through a \emph{flow function}. On the other hand, the way the inflow at an intermediate node gets split among its outgoing links depends on the current particle density, possibly on the whole network, through a \emph{routing policy}. Such a routing policy is said to be \emph{distributed} if the proportion of inflow routed to the outgoing links of a node is allowed to depend only on \emph{local information}, consisting of the current particle densities on the outgoing links of the same node.  
The inspiration for such a modeling paradigm comes from empirical findings from several application domains:   in transportation networks \cite{Garavello.Piccoli:06}, the flow functions are typically referred to as \emph{fundamental diagrams}, while the routing policies model the emerging selfish behavior of drivers; in data networks \cite{BertsekasGallager}, flow functions model congestion-dependent \emph{throughput} and \emph{average delays}, while routing policies are designed in order to optimize the total throughput or other performance measures; in production networks \cite{Karmarkar:89}, flow functions correspond to \emph{clearing functions}.

Our objective is the design and analysis of distributed routing policies for dynamical flow networks that are \emph{maximally robust} with respect to \emph{adversarial disturbances} that reduce the link flow capacities. Two notions of transfer efficiency are introduced in order to capture the extremes of the resilience of the network towards disturbances: The dynamical flow network is \emph{fully transferring} if the total inflow at the destination node asymptotically approaches the inflow at the origin node, and \emph{partially transferring} if the total inflow at the destination node is asymptotically bounded away from zero. The robustness of distributed routing policies is evaluated in terms of the network's \emph{strong} and \emph{weak} \emph{resilience}, which are defined as the minimum sum of link-wise magnitude of disturbances making the perturbed dynamical flow network not fully transferring, and, respectively, not partially transferring. In this paper, we prove that the maximum possible weak resilience is yielded by a class of \emph{locally responsive} distributed routing policies, which rely only on \emph{local information} on the current particle densities on the network, and are characterized by the property that the portion of its inflow that a node routes towards an outgoing link does not decrease as the particle density on any other outgoing link increases. Moreover, we show that the maximum weak resilience of dynamical flow networks with arbitrary, not necessarily distributed, routing policies equals the \emph{min-cut capacity} of the network and hence is independent of the initial equilibrium flow. We also prove some fundamental properties of dynamical flow networks driven by locally responsive distributed policies, including global convergence to a unique limit flow. Such properties are mainly a consequence of the particular \emph{cooperative} structure (in the sense of \cite{Hirsch:82,Hirsch:85}) that the dynamical flow network inherits from locally responsive routing policies. 

Stability analysis of network flow control policies under non-persistent disturbances, especially in the context of internet, has attracted a lot of attention, e.g., see \cite{Vinnicombe:02,Paganini:02,Low.Paganini.ea:02,Fan.Arcak.ea:04}. Recent work on robustness analysis of static flow networks under adversarial and probabilistic persistent disturbances in the spirit of this paper include  \cite{Vardi.Zhang:07,Bulteau.Rubino.ea:97,Sengoku.Shinoda.ea:88}. It is worth comparing the distributed routing policies studied in this paper with the back-pressure policy~\cite{Tassiulas.Ephremides:92}, which is one of the most well-known robust distributed routing policy for queueing networks. While relying on local information in the same way as the distributed routing policies studied here, back-pressure policies require the nodes to have, possibly unlimited, buffer capacity. In contrast, in our framework, the nodes have no buffer capacity. In fact, the distributed routing policies considered in this paper are closely related to the well-known \emph{hot-potato} or deflection routing policies \cite{Acampora.Shah:92} \cite[Sect.~5.1]{BertsekasGallager}, where the nodes route incoming packets immediately to one of the outgoing links. However, to the best of our knowledge, the robustness properties of dynamical flow networks, where the outflow from a link is not necessarily equal to its inflow have not been studied before. 

The contributions of this paper are as follows: (i) we formulate a novel dynamical system framework for robustness analysis of dynamical flow networks under feedback routing policies, possibly constrained in the available information; (ii) we characterize a general class of locally responsive distributed routing policies that yield the maximum weak resilience; (iii) we provide a simple characterization of the resilience in terms of the topology and capacity of the flow network. In particular, the class of locally responsive distributed routing policies can be interpreted as  approximate Nash equilibria in an appropriate zero-sum game setting where the objective of the adversary inflicting the disturbance is to make the network not partially transferring with a disturbance of minimum possible magnitude, and the objective of the system planner is to design distributed routing policies that yield the maximum possible resilience.  
The results of this paper imply that locality constraints on the information available to routing policies do not affect the maximally achievable weak resilience. In contrast,  the companion paper \cite{PartII} focuses on the strong resilience properties of dynamical flow networks, and shows that locally responsive distributed routing policies are maximally robust, but only within the class of distributed routing policies which are constrained to use only local information on the network congestion status.  


The rest of the paper is organized as follows. 
In Section~\ref{sec:model}, we formulate the problem by formally defining the notion of a dynamical flow network and its resilience, and we prove that the weak resilience of a dynamical flow network driven by an arbitrary, not necessarily distributed, routing policy is upper-bounded by the min-cut capacity of the network. In Section~\ref{sec:comparison}, we introduce the class of locally responsive distributed routing policies, and state the main results on dynamical flow networks driven by such locally responsive distributed routing policies: Theorem \ref{thm:uniquelimitflow}, concerning global convergence towards a unique equilibrium flow; and Theorem \ref{maintheo-weakstability} concerning the maximal weak resilience property. In Sections~\ref{sec:proofthmuniquelimit}, and \ref{sec:proof2}, we state proofs of Theorem \ref{thm:uniquelimitflow}, and Theorem \ref{maintheo-weakstability}, respectively.

Before proceeding, we define some preliminary notation to be used throughout the paper. Let $\real$ be the set of real numbers, $\realnoneg:=\{x\in\R:\,x\ge0\}$ be the set of nonnegative real numbers. Let $\mc A$ and $\mc B$ be finite sets. Then, $|\mc A|$ will denote the cardinality of $\mc A$, $\R^{\mc A}$ (respectively, $\R_+^{\mc A}$) the  space of real-valued (nonnegative-real-valued) vectors whose components are indexed  by elements of $\mc A$, and $\R^{\mc A\times\mc B}$ the space of matrices whose real entries indexed  by pairs of elements in $\mc A\times\mc B$. The transpose of a matrix $M \in \real^{\mc A \times\mc B}$, will be denoted by $M^T \in\R^{\mc B\times\mc A}$, while $\onebf$ the all-one vector, whose size will be clear from the context. Let $\cl(\mc X)$ be the closure of a set $\mc X\subseteq\R^{\mc A}$. A directed multigraph is the pair $(\mc V,\mc E)$ of a finite set $\mc V$ of nodes, and of a multiset $\mc E$ of links consisting of ordered pairs of nodes (i.e., we allow for parallel links). Given a a multigraph $(\mc V,\mc E)$, for every node $v\in\mc V$, we shall denote by $\mc E^+_v\subseteq\mc E$, and $\mc E^-_v\subseteq\mc E$, the set of its outgoing and incoming links, respectively. Moreover, we shall use the shorthand notation $\mc R_v:=\R_+^{\mc E^+_v}$ for the set of nonnegative-real-valued vectors whose entries are indexed  by elements of $\mc E^+_v$, $\mc S_v:=\{p\in \mc R_v:\,\sum_{e \in \mc E_v^+} p_e=1\}$ for the simplex of probability vectors over $\mc E^+_v$, and $\mc R:=\R_+^{\mc E}$ for the set of nonnegative-real-valued vectors whose entries are indexed  by the links in $\mc E$.

\section{Dynamical flow networks and their resilience}
\label{sec:model}
In this section, we introduce our model of dynamical flow networks and define the notions of transfer efficiency. 

\subsection{Dynamical flow networks}
We start with the following definition of a flow network.\medskip

\begin{definition}[Flow network]\label{def:flownetwork}
A \emph{flow network} $\mc N=(\mc T,\mu)$ is the pair of a \emph{topology}, described by a finite directed multigraph $\mc T=(\mc V,\mc E)$, where $\mc V$ is the node set and $\mc E$ is the link multiset, and a family of \emph{flow functions} $\mu:=\{\mu_e:\R_+\to\R_+\}_{e\in\mc E}$ describing the functional dependence $f_e=\mu_e(\rho_e)$ of the flow on the density of particles on every link $e\in\mc E$. 
The  \emph{flow capacity} of a link $e\in\mc E$ is defined as 
\be \flowmax_e:=\sup_{\rho_e \geq 0} \mu_e(\rho_e)\,. \ee
\end{definition}\medskip

\begin{figure}
\begin{center}
\includegraphics[width=10cm,height=7cm]{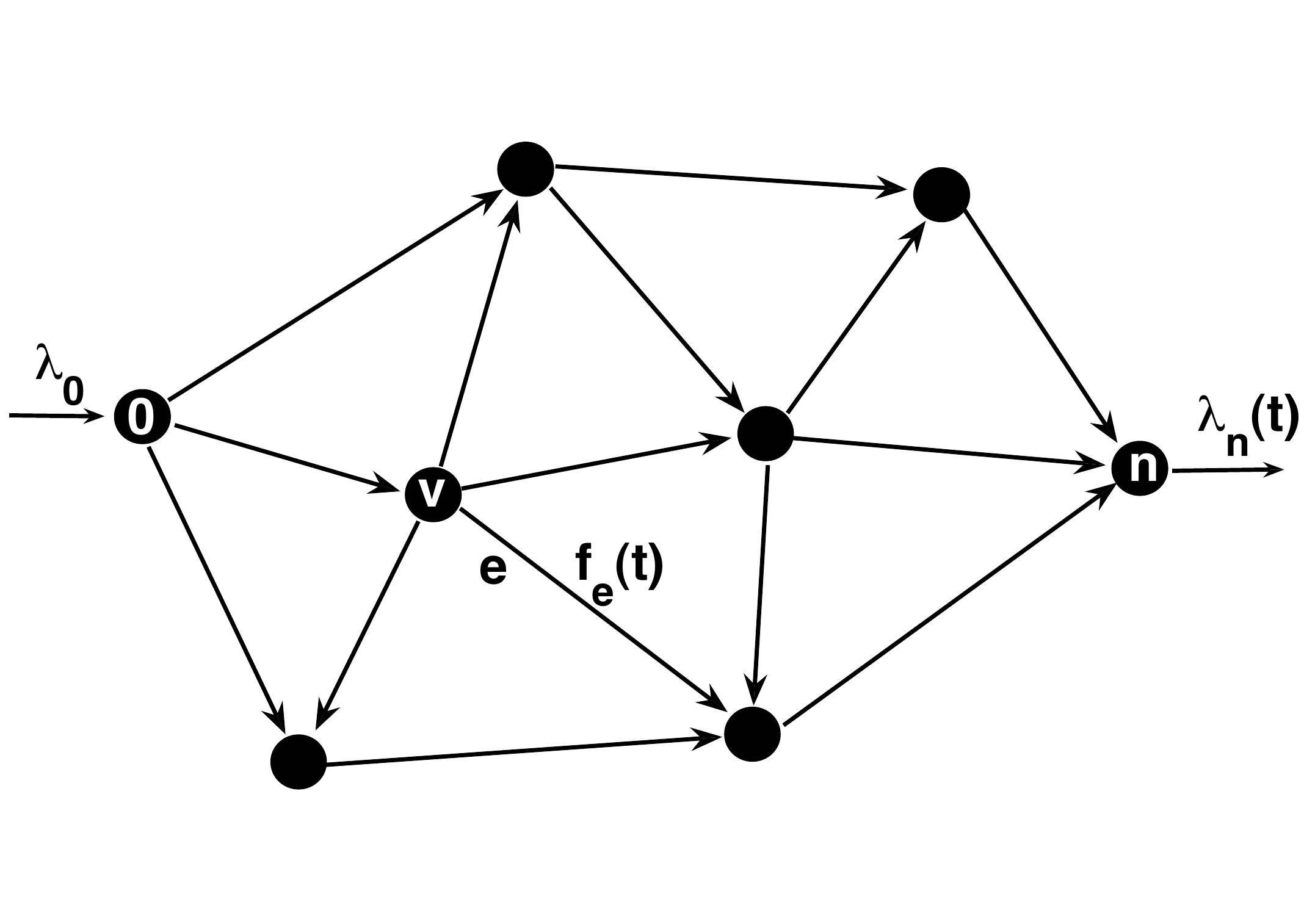}
\end{center}
\caption{\label{fig:IOnetwork}A network topology satisfying Assumption \ref{ass:acyclicity}: the nodes $v$ are labeled by the integers between $0$ (denoting the origin node) and $n$ (denoting the destination node), in such a way that the label of the head node of each edge is higher than the label of its tail node. The inflow at the origin, $\lambda_0$, maybe interpreted as the input to the dynamical flow network, and the total inflow at the destination, $\lambda_n(t)$, as the output. For $\alpha\in(0,1]$, the dynamical flow network is $\alpha$-transferring if $\liminf_{t\to+\infty}\lambda_n(t)\ge\alpha\lambda_0$, i.e., if at least $\alpha$-fraction of the inflow at the origin is transferred to the destination, asymptotically.}
\end{figure}

We shall use the notation $\mc F_v:=\times_{e\in\mc E^+_v}[0,\flowmax_e)$ for the set of admissible flow vectors on outgoing links from node $v$, and $\mc F:=\times_{e\in\mc E}[0,\flowmax_e)$ for the set of admissible flow vectors for the network. We shall write $f:=\{f_e:\,e\in\mc E\}\in\mc F$, and $\rho:=\{\rho_e:\,e\in\mc E\}\in\mc R$, for the vectors of flows and of densities, respectively, on the different links. The notation $f^v:=\{f_e:\,e\in\mc E^+_v\}\in\mc F_v$, and $\rho^v:=\{\rho_e:\,e\in\mc E^+_v\}\in\mc R_v$ will stand for the vectors of flows and densities, respectively, on the outgoing links of a node $v$. We shall compactly denote by $f=\mu(\rho)$ and $f^v=\mu^v(\rho^v)$ the functional relationships between density and flow vectors.

Throughout this paper, we shall restrict ourselves to network topologies satisfying the following: \medskip
\begin{assumption}\label{ass:acyclicity}
The topology $\mc T$ contains no cycles, has a unique origin (i.e., a node $v\in\mc V$ such that $\mc E^-_v$ is empty), and a unique destination (i.e., a node $v\in\mc V$ such that $\mc E^+_v$ is empty). Moreover, there exists a path in $\mc T$ to the destination node from every other node in $\mc V$. 
\end{assumption}\medskip

\begin{figure}
\begin{center}
\includegraphics[width=10cm,height=7cm]{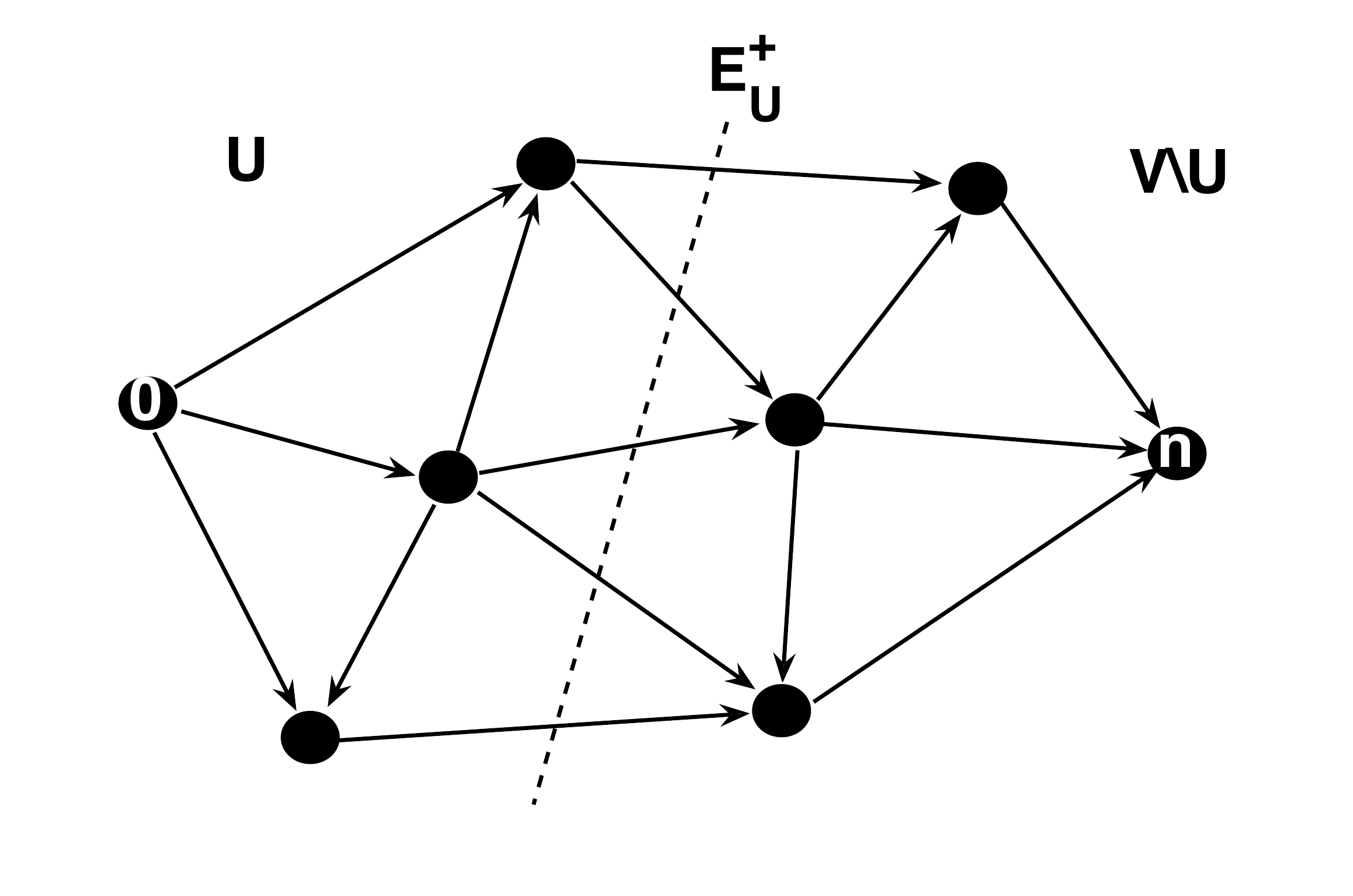}
\end{center}
\caption{An origin/destination cut of the network: $\mc U$ is a subset of nodes including the origin $0$ but not the destination $n$, and $\mc E^+_{\mc U}$ is the subset of those edges with tail node in $\mc U$, and head node in $\mc V\setminus\mc U$. \label{fig:mincut}}
\end{figure}

Assumption \ref{ass:acyclicity} implies that one can find a (not necessarily unique) topological ordering of the node set $\mc V$ (see, e.g., \cite{Cormen.Leiserson:90}). We shall assume to have fixed one such ordering, identifying $\mc V$ with the integer set $\{0,1,\ldots,n\}$, where $n:=|\mc V|-1$, in such a way that 
\be\label{vertexordering}\mc E^-_{v}\subseteq\bigcup\nolimits_{0\le u<v}\mc E^+_{u}\,,\qquad\forall  v=0,\ldots,n\,.\ee
In particular, (\ref{vertexordering}) implies that $0$ is the origin node, and $n$ the destination node in the network topology $\mc T$ (see Fig.~\ref{fig:IOnetwork}). An \emph{origin-destination cut}  (see, e.g., \cite{Ahuja.Magnanti.ea:93}) of $\graph$  is a partition of $\mc V$ into $\mc U$ and $\mc V \setminus\mc U$ such that $0 \in \mc U$ and $n \in \mc V \setminus \mc U$. Let \be\label{EU+def}\mc E_{\mc U}^+:=\{(u,v)\in\mc E:\,u\in\mc U,v\in\mc V\setminus\mc U\}\ee be the set of all the links pointing from some node in $\mc U$ to some node in $\mc V \setminus \mc U$ (see Fig.~\ref{fig:mincut}). The \emph{min-cut capacity} of a flow network $\mc N$ is defined as  
\be\label{def:capacity} C (\mc N):=\min_{\mc U} \sum\nolimits_{e \in \mc E_{\mc U}^+} \flowmaxe\,,\ee
where the minimization runs over all the origin-destination cuts of $\mc T$. Throughout this paper, we shall assume a constant inflow $\lambda_0 \ge 0$ at the origin node.
Let us define the set of \emph{admissible equilibrium flows} associated to an inflow $\lambda_0$ as  
$$\mc F^*(\lambda_0):=\l\{f^*\in\mc F:\,\sum\nolimits_{e\in\mc E^+_0}f_e^*=\lambda_0,\,\sum\nolimits_{e\in\mc E^+_v}f_e^*=\sum\nolimits_{e\in\mc E^-_v}f_e^*,\,\forall \, 0<v<n\r\}\,.$$
Then, it follows from the max-flow min-cut theorem (see, e.g., \cite{Ahuja.Magnanti.ea:93}), that $\mc F^*(\lambda_0)\ne\emptyset$ whenever $\lambda_0<C(\mc N)$. That is, the min-cut capacity equals the maximum flow that can pass from the origin to the destination node while satisfying capacity constraints on the links, and conservation of mass at the intermediate nodes.

\begin{figure}
\begin{center}
\includegraphics[width=10cm,height=6cm]{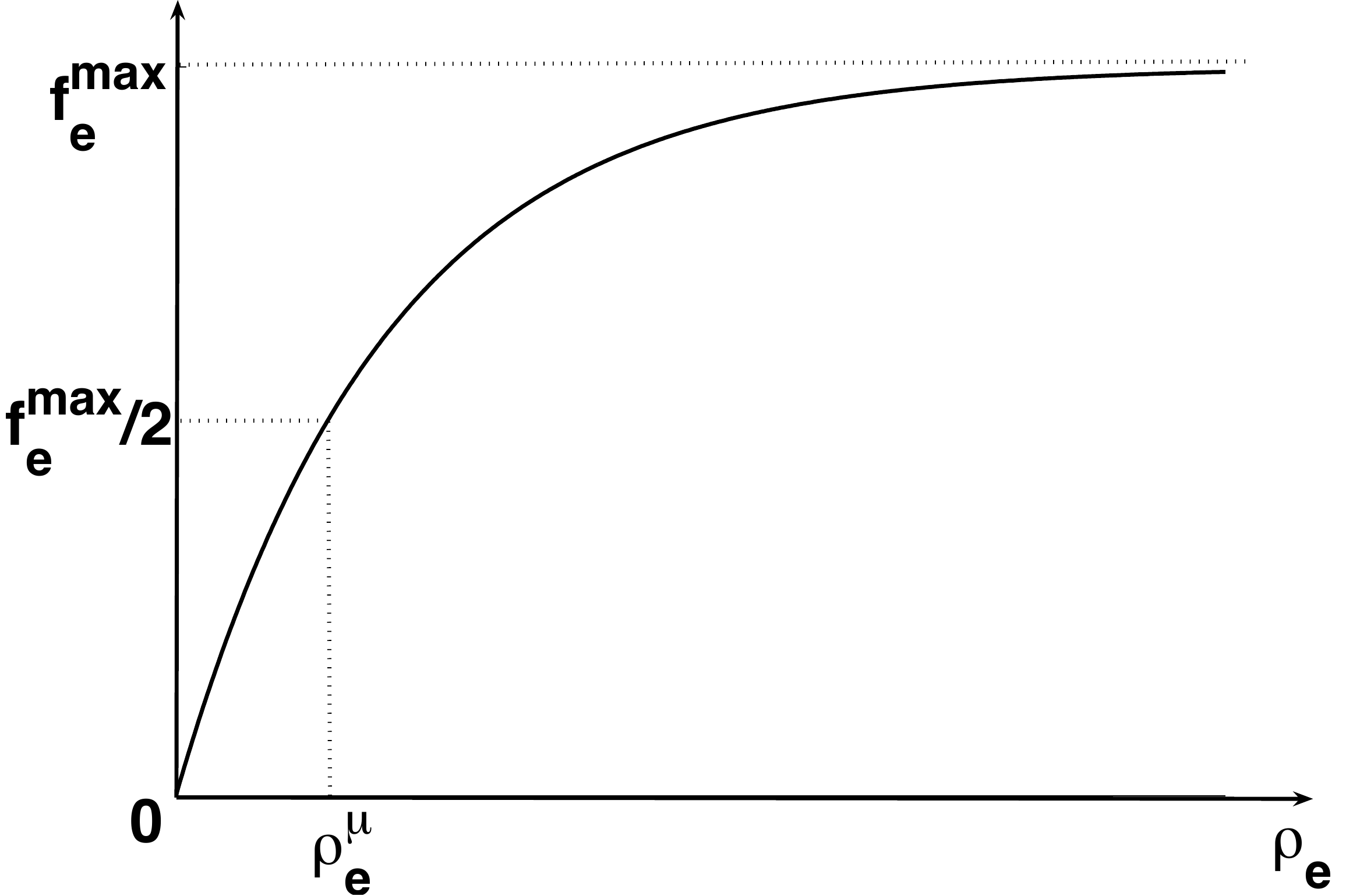}
\end{center}
\caption{\label{fig:flowfunction}Qualitative behavior of a flow function satisfying Assumption \ref{ass:flowfunction}: $\mu_e(\rho_e)$ is differentiable, strictly increasing, has bounded derivative and such that $\mu_e(0)=0$, and $\lim\limits_{\rho_e\to+\infty}\mu_e(\rho_e)=f_e^{\max}<+\infty$. The median density $\rho^{\mu}_e$, as defined in (\ref{eq:median-def}) is plotted as well. }
\end{figure}

Throughout the paper, we shall make the following assumption on the flow functions (see also Fig.~\ref{fig:flowfunction}): \medskip
\begin{assumption}\label{ass:flowfunction}
For every link $e\in\mc E$, the map $\mu_e:\R_+\to\R_+$ is continuously differentiable, strictly increasing, has bounded derivative, and is such that $\mu_e(0)=0$, and $f_e^{\max}<+\infty$. 
\end{assumption}\medskip 

Thanks to Assumption \ref{ass:flowfunction}, one can define the \emph{median density} on link $e\in\mc E$ as the unique value $\rhomedian_e\in\R_+$ such that 
\be
\label{eq:median-def}
\mu_e(\rhomedian_e)=f_e^{\max}/2.
\ee

\begin{example}[Flow function]
\label{example:flowfunction}
For every link $e\in\mc E$, let $a_e$ and $f_e^{\max}$ be positive real constants. Then, a simple example of flow function satisfying Assumption~\ref{ass:flowfunction} is given by $$\mu_e(\rho_e)=\flowmax_e \left(1- \exp(-a_e \rho_e) \right)\,.$$ It is easily verified that the flow capacity is $\flowmax_e$, while the median density for such a flow function is $\rhomedian_e=a_e^{-1}\log2$.
\end{example}\medskip

We now introduce the notion of a distributed routing policy used in this paper. 

\begin{definition}[(Distributed) routing policy]\label{def:distributedroutingpolicy}
A \emph{routing policy} for a flow network $\mc N$ is a family of differentiable functions $\mc G:=\{G^v:\mc R\to\mc S_v \}_{0\le v<n}$ describing the ratio in which the particle flow incoming in each non-destination node $v$ gets split among its outgoing link set $\mc E^+_v$, as a function of the observed current particle density. A routing policy is said to be \emph{distributed} if, for all $0\le v<n$, there exists a differentiable function $\ov G:\mc R_v\to\mc S_v$ such that $G^v(\rho)=\ov G^v(\rho^v)$ for all $\rho\in\mc R$, where $\rho^v$ is the projection of $\rho$ on the outgoing link set $\mc E^+_v$. 
\end{definition}\medskip

The salient feature in Definition \ref{def:distributedroutingpolicy} is that a distributed routing policy depends only on the \emph{local information} on the particle density $\rho^v$ on the set $\mc E^+_v$ of outgoing links of the non-destination node $v$, instead of the full vector of current particle densities $\rho$ on the whole link set $\mc E$. Throughout this paper, we shall make a slight abuse of notation and write $G^v(\rho^v)$, instead of $\ov G^v(\rho^v)$, for the vector of the fractions in which the inflow of node $v$ gets split into its outgoing links. 

We are now ready to define a dynamical flow network.\medskip 

\begin{definition}[Dynamical flow network]\label{def:dynamicalflownetwork}
A \emph{dynamical flow network} associated to a flow network $\mc N$ satisfying Assumption \ref{ass:acyclicity}, a distributed routing policy $\mc G$, and an inflow $\lambda_0\ge0$, is the dynamical system 
\be\label{dynsyst}
\ds\frac{\de}{\de t}\rho_e(t)=\lambda_v(t)G^v_e(\rho(t))-f_e(t)\,,\qquad \forall\,0\le v<n\,,\quad\forall\,e\in\mc E^+_v\,,\ee
where
\be\label{felambdadef} f_e(t):=\mu_e(\rho_e(t))\,,\qquad \lambda_v(t):=\l\{\ba{lcl}\lambda_0&\se&v=0\\\sum_{e\in\mc E^-_v}f_e(t)&\se&0<v \leq  n.\ea\r.\ee
\end{definition}\medskip

Equation (\ref{dynsyst}) states that the rate of variation of the particle density on a link $e$ outgoing from some non-destination node $v$ is given by the difference between $\lambda_v(t)G^v_e(\rho(t))$, i.e., the portion of the inflow at node $v$ which is routed to link $e$, and $f_e(t)$, i.e., the particle flow on link $e$.  Observe that the (distributed) routing policy $G^v(\rho)$ induces a (local) feedback which couples the dynamics of the particle flow on the the different links. 

We can now introduce the following notion of transfer efficiency of a dynamical flow network.\medskip  
\begin{definition}[Transfer efficiency of a dynamical flow network]\label{def:alphatransferring}
Consider a dynamical flow network $\mc N$ satisfying Assumptions \ref{ass:acyclicity} and \ref{ass:flowfunction}.
Given some flow vector $f^{\circ}\in\mc F$, and $\alpha\in[0,1]$, the dynamical flow network (\ref{dynsyst}) is said to be  \emph{$\alpha$-transferring} with respect to $f^{\circ}$ if the solution of (\ref{dynsyst}) with initial condition $\rho(0)=\mu^{-1}(f^{\circ})$ satisfies 
\be\label{alphatransferringdef}\liminf_{t\to+\infty}\lambda_n(t)\ge\alpha\lambda_0\,.\ee
\end{definition}
\medskip

Definition \ref{def:alphatransferring} states that a dynamical flow network is $\alpha$-transferring when the outflow is asymptotically not smaller than $\alpha$ times the inflow. In particular, a fully transferring dynamical flow network is characterized by the property of having outflow asymptotically equal to its inflow, so that there is no throughput loss. On the other hand, a partially transferring dynamical flow network might allow for some throughput loss, provided that some fraction of the flow is still guaranteed to be asymptotically transferred.

\begin{remark}\label{remark:standardflow}
Standard definitions in the literature are typically limited to static flow networks describing the particle flow at equilibrium via conservation of mass. In fact, they usually consist (see e.g., \cite{Ahuja.Magnanti.ea:93}) in the specification of a topology $\mc T$, a vector of flow capacities $f^{\max}\in\mc R$, and an admissible equilibrium flow vector $f^*\in\mc F^*(\lambda_0)$ for $\lambda_0<C(\mc N)$ (or, often, $f^*\in\cl(\mc F^*(\lambda_0))$ for $\lambda_0\le C(\mc N)$). In contrast, in our model, we focus on the off-equilibrium particle dynamics on a flow network $\mc N$, induced by a (distributed) routing policy $\mc G$. 
\end{remark}


\subsection{Examples}\label{sec:examples}

We now present three illustrative applications of the dynamical flow network framework. 

\begin{enumerate}
\item \textit{Transportation networks}:
In transportation networks, particles represent drivers and distributed routing policies correspond to their local route choice behavior in response to the locally observed link congestions. A desired route choice behavior from a social optimization perspective may be achieved by appropriate incentive mechanisms. While we do not address the issue of mechanism design in this paper,  the companion paper \cite{PartII} discusses the use of tolls in influencing the long-term global route choice behavior of drivers to get a desired initial equilibrium state for the network. The robust distributed routing policies designed in this paper would correspond to the \emph{ideal} node-wise route choice behavior of the drivers.
The flow function $\mu_e(\rho_e)$ presented in this paper is related to the notion of fundamental diagram in traffic theory, e.g., see \cite{Garavello.Piccoli:06}. Note that in our formulation, we assume that the density of drivers is homogeneous over a link. One can refer to \cite{Garavello.Piccoli:06} for models that incorporate inhomogeneity, although such models are developed under non-feedback routing policies. 

\item \textit{Data networks}:
In data networks, the particles represent data packets that are to be routed from sources to destinations by routers placed at the nodes (see, e.g., \cite[Ch.~5]{BertsekasGallager}). Typically the average packet delay from one router to the other increases with the increase in queue length on the link between the two routers. Hence, one has that such average delay is given by $d_e(\rho_e)$, where $d_e(\rho_e)$ is an increasing function. If one further assumes that the delay function $d_e(\rho_e)$ is concave and such that $\liminf_{\rho_e\to+\infty}d_e(\rho_e)/\rho_e>0$, then the relationship between the throughput and the queue length, $f_e\propto\rho_e/d_e(\rho_e)$, can be easily shown to satisfy Assumption \ref{ass:flowfunction}. Therefore, in analogy with the general framework, $\rho_e$ and $f_e$ denote the queue length and the throughput, respectively, and $\mu_e(\rho_e)$ represents the throughput functions on the links of data networks.


\item \textit{Production networks}:
In production networks, the particles represent goods that need to be processed by a series of production modules represented by nodes. It is known, e.g., see \cite{Karmarkar:89}, that the rate of doing work decreases with the amount of work in progress at a production module. This relationship is formalized by the concept of \emph{clearing functions}.
In this context, production networks have a clear analogy with our setup where $\rho_e$ represents the work-in-progress, $f_e$ represents the rate of doing work, and $\mu_e(\rho_e)$ represents the clearing function.

\end{enumerate}

\begin{remark}
While there are many examples of congestion-dependent throughput functions and clearing functions that satisfy Assumption~\ref{ass:flowfunction}, typical fundamental diagrams in transportation systems have a $\cap$-shaped profile. While we do not study the implications of this analytically, some simulations are provided in \cite{PartII} illustrating how the results of this paper could be extended to this case.
%
\end{remark}

\begin{remark}\label{remark2timescales}
It is worth stressing that, while distributed routing policies depend only on local information on the current congestion, their structural form may depend on some global information on the flow network which might have been accumulated through a slower time-scale evolutionary dynamics. A two time-scale process of this sort has been analyzed in our related work \cite{Como.Savla.ea:Wardrop-arxiv} in the context of transportation networks. Multiple time-scale dynamical processes have also been analyzed in \cite{Borkar.Kumar:03} in the context of communication networks.  
\end{remark}

\subsection{Perturbed dynamical flow networks and resilience} \label{sec:perturbations}
We shall consider persistent perturbations of the dynamical flow network (\ref{dynsyst}) that reduce the flow functions on the links, as per the following: 
\begin{definition}[Admissible perturbation]\label{def:admissibeperturbation}
An \emph{admissible perturbation} of a flow network $\mc N=(\mc T,\mu)$, satisfying Assumptions \ref{ass:acyclicity} and \ref{ass:flowfunction}, is a flow network $\tilde{\mc N}=({\mc T},\tilde\mu)$, with the same topology $\mc T$, and a family of perturbed flow functions $\tilde\mu:=\{\tilde\mu_e:\R_+\to\R_+\}_{e\in\mc E}$, such that, for every $e\in\mc E$, $\tilde\mu_e$ satisfies Assumption \ref{ass:flowfunction}, as well as
$$\tilde\mu_e(\rho_e)\le\mu_e(\rho_e)\,,\qquad \forall \rho_e\ge0\,.$$
We accordingly let $\newflowmaxe:=\sup\{\tilde{\mu}_e(\tilde{\rho}_e):\tilde{\rho}_e\ge0\}$.
The \emph{magnitude} of an admissible perturbation is defined as 
\be\label{deltadef}\delta:=\sum\nolimits_{e\in\mc E}\delta_e\,,\qquad\delta_e:=\sup\l\{\mu_e(\rho_e)-\tilde\mu_e(\rho_e):\,\rho_e\ge0\r\}\,.\ee 
The \emph{stretching coefficient} of an admissible perturbation is defined as 
\be\label{thetadef}\theta:=\max\{\tilderhomedian_e/\rhomedian_e:\,e\in\mc E\}\,,\ee 
where $\rhomedian_e$, and $\tilde{\rho}^{\mu}_e$ are the median densities associated to the unperturbed and the perturbed flow functions, respectively, on link $e\in\mc E$, as defined in \eqref{eq:median-def}. 
\end{definition}\medskip

Given a dynamical flow network as in Definition \ref{def:dynamicalflownetwork}, and an admissible perturbation as in Definition \ref{def:admissibeperturbation}, we shall consider the \emph{perturbed dynamical flow network}
\be\label{pertdynsyst}
\ds\frac{\de}{\de t}\tilde\rho_e(t)=\tilde\lambda_v(t)G^v_e(\tilde\rho(t))- \tilde f_e(t)\,,\qquad\forall\,0\le v<n\,,\quad\forall\,e\in\mc E^+_v\,,\ee
where
\be\tilde f_e(t):=\tilde\mu_e(\tilde\rho_e(t))\,,\qquad \tilde\lambda_v(t):=\l\{\ba{lcl}\sum_{e\in\mc E^-_v}\tilde f_e(t)&\se&0<v < n\\\lambda_0&\se&v=0\,.\ea\r.
\ee
Observe that the perturbed dynamical flow network (\ref{pertdynsyst}) has the same structure of the original dynamical flow network (\ref{dynsyst}), as it describes the rate of variation of the particle density on each link $e$ outgoing from some non-destination node $v$ as the difference between $\tilde\lambda_v(t)G^v_e(\tilde\rho(t))$, i.e., the portion of the perturbed inflow at node $v$ which is routed to link $e$, minus the perturbed flow on link $e$ itself. Notice that the only difference with respect to the original dynamical flow network (\ref{dynsyst}) is in the perturbed flow function $\tilde\mu_e(\rho_e)$ on each link $e\in\mc E$, which replaces the original one, $\mu_e(\rho_e)$. In particular, the (distributed) routing policy $\mc G$ is the same for the unperturbed and the perturbed dynamical flow networks. In this way, we model a situation in which the routers are not aware of the fact that the flow network has been perturbed, but react to this change only indirectly, in response to variations of the local density vectors $\tilde\rho^v(t)$. 

We are now ready to define the following notion of resilience of a dynamical flow network as in Definition \ref{def:dynamicalflownetwork} with respect to an initial flow.
\medskip
\begin{definition}[Resilience of a dynamical flow network]\label{def:stabilitymargins}
Let $\mc N$ be a flow network satisfying Assumptions \ref{ass:acyclicity} and \ref{ass:flowfunction}, $\mc G$ be a distributed routing policy, and $\lambda_0\ge0$ be a constant inflow at the origin node. Given $\alpha\in(0,1]$, $\theta\ge1$ and $f^{\circ} \in \mc F$, let $\gamma_{\alpha,\theta}(f^{\circ},\mc G)$ be equal to the infimum magnitude of all the admissible perturbations of stretching coefficient less than or equal to $\theta$ for which the perturbed dynamical flow network (\ref{pertdynsyst}) is not $\alpha$-transferring with respect to $f^{\circ}$. Also, define $$\gamma_{0,\theta}(f^{\circ},\mc G):=\lim_{\alpha\downarrow0}\gamma_{\alpha,\theta}(f^{\circ},\mc G)\,.$$ For $\alpha\in[0,1]$, the \emph{$\alpha$-resilience} with respect to $f^{\circ}$ is defined as\footnote{It is easily seen that the limits involved in this definition always exist, as $\gamma_{\alpha,\theta}(f^{\circ},\mc G)$ is clearly nonincreasing in $\alpha$ (the higher $\alpha$, the more stringent the requirement of $\alpha$-transfer) and $\theta$ (the higher $\theta$, the more admissible perturbations are considered that may potentially make the dynamical flow network to be not $\alpha$-transferring).} 
$$\gamma_{\alpha}(f^{\circ},\mc G):=\lim_{\theta\to+\infty}\gamma_{\alpha,\theta}(f^{\circ},\mc G)\,.$$ The $1$-resilience will be referred to as the \emph{strong resilience}, while the $0$-resilience will be referred to as the \emph{weak resilience}.
 \end{definition}\medskip
 \begin{remark}[Zero-sum game interpretation]
 The notions of resilience are with respect to adversarial perturbations. Therefore, one can provide a zero-sum game interpretation as follows. Let the strategy space of the system planner be the class of distributed routing policies and the strategy space of an adversary be the set of admissible perturbations. Let the utility function of the adversary be $M \Theta - \delta$, where $M$ is a large quantity, e.g., $\sum_{e \in \mc E} \flowmaxe$, and $\Theta$ takes the value $1$ if the network is not $\alpha$-transferring under given strategies of the system planner and the adversary, and zero otherwise. Let the utility function of the system planner be $\delta - M \Theta$. As stated in Sectio	n \ref{sec:comparison}, a certain class of \emph{locally responsive} distributed routing policies characterized by Definition \ref{def:myopicpolicy}, is maximally robust with respect to the notions of weak and strong resilience. This will then show that the locally responsive distributed routing policies correspond to approximate Nash equilibria in this zero-sum game setting.
 \end{remark}\medskip 
 
In the remainder of the paper, we shall focus on the characterization of the weak resilience of dynamical flow networks, while the strong resilience will be addressed in the companion paper \cite{PartII}. Before proceeding, let us elaborate a bit on Definition \ref{def:stabilitymargins}. Notice that, for every $\alpha\in(0,1]$, the $\alpha$-resilience $\gamma_{\alpha}(f^{\circ},\mc G)$ is simply the infimum magnitude of all the admissible perturbations such that the perturbed dynamical network (\ref{pertdynsyst}) is not $\alpha$-transferring with respect to the equilibrium flow $f^{\circ}$. In fact, one might think of $\gamma_{\alpha}(f^{\circ},\mc G)$ as the minimum effort required by a hypothetical adversary in order to modify the dynamical flow network  from (\ref{dynsyst}) to (\ref{pertdynsyst}), and make it not $\alpha$-transferring, provided that such an effort is measured in terms of the magnitude of the perturbation $\delta=\sum_{e \in \mc E} ||\mu_e(\,\cdot\,)-\tilde\mu_e(\,\cdot\,)||_{\infty}$. For $\alpha=0$, trivially the perturbed network flow is always $0$-transferring with respect to any initial flow. For this reason, the definition of the weak resilience $\gamma_0(f^{\circ},\mc G)$ involves the double limit $\lim_{\theta\to+\infty}\lim_{\alpha\downarrow0}\gamma_{\alpha,\theta}(f^{\circ},\mc G)$: the introduction of the bound on the stretching coefficient of the admissible perturbation is a mere technicality whose necessity will become clear in Section \ref{sec:proof2}. 

We conclude this section with the following result, providing an upper bound on the weak resilience of a dynamical flow network driven by any, not necessarily distributed, routing policy $\mc G$, in terms of the min-cut capacity of the network. Tightness of this bound will follow from Theorem \ref{maintheo-weakstability} in Section \ref{sec:comparison}, which will show that, for a particular class of locally responsive distributed routing policies, the dynamical flow network has weak resilience equal to the min-cut capacity. 

\begin{proposition}\label{propUB}
Let $\mc N$ be a flow network satisfying Assumptions \ref{ass:acyclicity} and \ref{ass:flowfunction}, $\lambda_0>0$ a constant inflow, and $\mc G$ an arbitrary routing policy. Then, for any initial flow $f^{\circ}$, the weak resilience of the associated dynamical flow network satisfies 
$$\gamma_0(f^{\circ},\mc G)\le C(\mc N)\,.$$
\end{proposition} 
\begin{proof}
We shall prove that, for every $\alpha\in(0,1]$, and every $\theta\ge1$, 
\be\label{gammaalphatheta}\gamma_{\alpha,\theta}(f^{\circ},\mc G)\le C(\mc N)-\frac{\alpha}2\lambda_0\,.\ee 
Observe that (\ref{gammaalphatheta}) immediately implies that $$\gamma_0(f^{\circ},\mc G)=\lim_{\theta\to+\infty}\lim_{\alpha\downarrow0}\gamma_{\alpha,\theta}(f^{\circ},\mc G)\le \lim_{\theta\to+\infty}\lim_{\alpha\downarrow0} \left (C(\mc N)-\alpha\lambda_0/2 \right)=C(\mc N)\,,$$ 
thus proving the claim. 

Consider a minimal origin-destination cut, i.e., some $\mc U\subseteq\mc V$ such that $0\in\mc U$, $n\notin\mc U$, and $\sum_{e\in\mc E^+_{\mc U}}f^{\max}_e=C(\mc N)$. Define $\eps:=\alpha\lambda_0/(2C(\mc N))$, and consider an admissible perturbation such that $\tilde\mu_e(\rho_e)=\eps\mu_e(\rho_e)$ for every $e\in\mc E_{\mc U}^+$, and $\tilde\mu_e(\rho_e)=\mu_e(\rho_e)$ for all $e\in\mc E\setminus\mc E_{\mc U}^+$. It is readily verified that the magnitude of such perturbation satisfies 
$$\delta=(1-\eps)\sum\nolimits_{e\in\mc E^+_{\mc U}}f_e^{\max}=(1-\eps)C(\mc N)=C(\mc N)-\frac{\alpha}2\lambda_0\,,$$ 
while its stretching coefficient is $1$. 

Observe that 
\be\label{lambdaUbound} \tilde \lambda_{\mc U}(t):= \sum_{e \in \mc E_{\mc U}^+} \tilde f_e(t) \le\sum\nolimits_{e\in\mc E^+_{\mc U}}\tilde f_e^{\max}=\eps\sum\nolimits_{e\in\mc E^+_{\mc U}}f_e^{\max}=\alpha\lambda_0/2\,,\qquad t\ge0\,.\ee
Now, let $\mc W:=\mc V\setminus\mc U$ be the set of nodes on the destination side of the cut, and observe that 
\be\label{derhotiledet}\ba{rcl}
\ds\frac{\de}{\de t}\l(\sum\nolimits_{e\in\mc E^+_w}\tilde\rho_e(t)\r)&=&
\ds\sum\nolimits_{e\in\mc E^+_w}\l(\sum\nolimits_{j\in\mc E^-_w}\tilde f_j(t)\r)G^v_e(\tilde\rho(t))-\sum\nolimits_{e\in\mc E^+_w}\tilde f_e(t)\\[10pt]
&=&\ds\sum\nolimits_{e\in\mc E^-_w}\tilde f_e(t)-\sum\nolimits_{e\in\mc E^+_w}\tilde f_e(t)
\ea
\ee
Define $\mc A:=\cup_{w\in\mc W}\mc E^+_w$, $\mc B:=\cup_{w\in\mc W}\mc E^-_w$, and let $\zeta(t):=\sum_{e\in\mc A}\tilde \rho_e(t)$. From (\ref{derhotiledet}), the identity $\mc A\cup\mc E^+_{\mc U}=\mc B$, and (\ref{lambdaUbound}), one gets
\be\label{dezetadet}\ba{rcl}\ds\frac{\de}{\de t}\zeta(t)
&=&\ds\sum\nolimits_{w\in\mc W}\sum\nolimits_{e\in\mc E^+_w}\frac{\de}{\de t}\tilde \rho_e(t)
\\[10pt]&=&
\ds\sum\nolimits_{e\in\mc B}\tilde f_e(t)-\sum\nolimits_{e\in\mc E^-_n}\tilde f_e(t)-\sum\nolimits_{e\in\mc A}\tilde f_e(t)\\[10pt]&=&\ds\sum\nolimits_{e\in\mc E^+_{\mc U}}\tilde f_e(t)-\sum\nolimits_{e\in\mc E^-_n}\tilde f_e(t)\\[10pt]&<&\ds \alpha\lambda_0/2-\tilde\lambda_n(t)\,.
\ea\ee
Now assume, by contradiction, that 
$$\liminf_{t\to+\infty}\tilde\lambda_n(t)\ge\alpha\lambda_0\,.$$
Then, there would exist some $\tau\ge0$ such that $\tilde\lambda_n(t)\ge 3\alpha\lambda_0/4$ for all $t\ge\tau$. For all $t\ge\tau$, it would then follow from (\ref{dezetadet}) that $\de\zeta(t)/\de t\le-\alpha\lambda_0/4\,,$ so that  
$$\zeta(t)\le\zeta(\tau)+(t-\tau)\alpha\lambda_0/4$$
by Gronwall's inequality. Therefore, $\zeta(t)$ would converge to $-\infty$ as $t$ grows large, contradicting the fact that $\zeta(t)\ge0$ for all $t\ge0$.
Then, necessarily $$\liminf_{t\to+\infty}\tilde\lambda_n(t)<\alpha\lambda_0\,,$$ so that the perturbed dynamical network is not $\alpha$-transferring. This implies (\ref{gammaalphatheta}), and therefore the claim.  \end{proof}

\section{Main results and discussion} 
\label{sec:comparison}

In this paper, we shall be concerned with a family of \emph{maximally robust} distributed routing policies. Such a family is characterized by the following:
\begin{definition}[Locally responsive distributed routing policy]\label{def:myopicpolicy}
A \emph{locally responsive} distributed routing policy for a flow network topology $\mc T=(\mc V,\mc E)$ with node set $\mc V=\{0,1,\ldots,n\}$ is a family of continuously differentiable distributed routing functions $\mc G=\{G^v:\mc R_v\to\mc S_v\}_{v\in\mc V}$ such that, for every non-destination node $0\le v<n$:
\begin{description}
\item[(a)]
$\ds\frac{\partial}{\partial \rho_e}G^v_j(\rho^v)\ge0 \,,\qquad \forall j,e\in\mc E^+_v\,,  j\ne e\,,\rho^v\in\mc R^v\,;$
\item[(b)]
for every nonempty proper subset $\mc J\subsetneq\mc E^+_v$, there exists a continuously differentiable map $G^{\mc J}:\mc R_{\mc J}\to\mc S_{\mc J}$, where $\mc R_{\mc J}:=\R_+^{\mc J}$, and $\mc S_{\mc J}:=\{p\in\mc R_{\mc J}:\,\sum_{j\in\mc J}p_j=1\}$ is the simplex of probability vectors over $\mc J$, such that, for every $\rho^{\mc J}\in\mc R_{\mc J}$, if $$\rho^v_e\to+\infty\,,\ \ \forall e\in\mc E^+_v\setminus\mc J\,,\qquad\rho^v_j\to\rho_j^{\mc J}\,,\ \ \forall j\in\mc J\,,$$ then  $$G^v_e(\rho^v)\to0,\ \ \forall e\in\mc E^+_v\setminus\mc J\,,\qquad G^v_j(\rho^v)\to G^{\mc J}_j(\rho^{\mc J}),\ \ \forall j\in\mc J\,.$$
\end{description}
\end{definition}

Property (a) in Definition \ref{def:myopicpolicy} states that, as the particle density on an outgoing link $e\in\mc E^+_v$ increases while the particle density on all the other outgoing links remains constant, the fraction of inflow at node $v$ routed to any link $j\in\mc E^+_v\setminus\{e\}$ does not decrease, and hence the fraction of inflow routed to link $e$ itself does not increase. In fact, Property (a) in Definition~\ref{def:myopicpolicy} is reminiscent of Hirsch's notion of \emph{cooperative dynamical systems} \cite{Hirsch:82,Hirsch:85}. On the other hand, Property (b) implies that the fraction of incoming particle flow routed to a subset of outgoing links $\mc K\subset\mc E^+_v$ vanishes as the density on links in $\mc K$ grows unbounded while the density on the remaining outgoing links remains bounded. It is worth observing that, when the routing policy models some selfish behavior of the particles (e.g., in transportation networks), then Property (a) and (b) are very natural assumptions on such behavior as they capture some sort of greedy local minimization of the delay.

\begin{example}[Locally responsive distributed routing policy]
\label{example:routing}
Let $\ilogitconst_v$, for $0\le v<n$, and $a_e$, for $e\in\mc E$, be positive constants. Define the routing policy $\mc G$ by
\begin{equation}
\label{eq:routing-example}
G^v_e(\rho)=\frac{a_e \exp(- \ilogitconst_v\rho_e)}{\sum_{j \in \mc E_v^+} a_j \exp(- \ilogitconst_v \rho_j)}\,, \qquad \forall e \in \mc E_v^+\,, \quad \forall 0\le v<n\,.
\end{equation} Clearly, $\mc G$ is distributed, as it uses only information on the particle density on the links outgoing from a node $v$ in order to compute how the inflow at node $v$ gets split among its outgoing links. Moreover, for all $0\le v<n$, and $e\in\mc E^+_v$, $G^v_e(\rho)$ is clearly differentiable, and computing partial derivatives one gets 
\begin{equation}
\label{eq:routing-example-partialder}
\frac{\partial}{\partial\rho_j}G^v_e(\rho) =\ilogitconst_v \frac{a_e a_j \exp(-\ilogitconst_v\rho_e) \exp(-\ilogitconst_v\rho_j)}{\left(\sum_{i \in \mc E_v^+} \alpha_i \exp(- \ilogitconst_v\rho_i) \right)^2}\ge0 \qquad \forall j \in \mc E_v^+, \quad j \neq e\,,
\end{equation}
and $\frac{\partial}{\partial\rho_j}G_e(\rho) =0$ for all $j\in\mc E\setminus\mc E^+_v$.
This implies that Property (a) of Definition~\ref{def:myopicpolicy} holds true. Property (b) is also easily verified. Therefore, $\mc G$ is a locally responsive distributed routing policy. In the context of transportation networks, the example in \eqref{eq:routing-example} is a variant of the logit function from discrete choice theory emerging from utilization maximization perspective of drivers, where the utility associated with link $e$ is the sum of $-\rho_e+\log a_e/\eta_v$ and a double exponential random variable with parameter $\eta_v$ (see, e.g., \cite{BenAkiva.Lerman:85}). 
\end{example}\medskip

We are now ready to state our main results. The first one shows that, when the distributed routing policy $\mc G$ is locally responsive, the dynamical flow network (\ref{dynsyst}) always admits a unique, globally attractive limit flow vector. 
\smallskip
\begin{theorem}[Existence of a globally attractive limit flow under locally responsive routing policies]\label{thm:uniquelimitflow}
Let $\mc N$ be a flow network satisfying Assumptions \ref{ass:acyclicity} and \ref{ass:flowfunction}, $\lambda_0\ge0$ a constant inflow, and $\mc G$ a locally responsive distributed routing policy. Then, there exists a unique limit flow $f^*\in\cl(\mc F)$ such that, for every initial condition $\rho(0)\in\mc R$, the dynamical flow network (\ref{dynsyst}) satisfies 
$$\lim_{t\to+\infty}f(t)=f^*\,.$$
Moreover, the limit flow $f^*$ is such that, if $f_e^*=f_e^{\max}$ for some link $e\in\mc E^+_v$ outgoing from a nondestination node $0\le v<n$, then $f_e^*=f_e^{\max}$ for every outgoing link $e\in\mc E^+_v$. 
\end{theorem}
\begin{proof}
See Section \ref{sec:proofthmuniquelimit}.
\end{proof} \medskip

Theorem \ref{thm:uniquelimitflow} states that, when the routing policy is distributed and locally responsive, there is a unique globally attractive limit flow $f^*$. Such a limit flow may be in $\mc F$, in which case it is not hard to see that it is necessarily an equilibrium flow, i.e., $f^*\in\mc F^*(\lambda_0)$; or belong to $\cl(\mc F)\setminus\mc F$, i.e., it satisfies the capacity constraint on one link with equality, in which case it is not an equilibrium flow. In the latter case, it satisfies the additional property that, on all the links outgoing from the same node, the capacity constraints are satisfied with equality. Such additional property will prove particularly useful in our companion paper \cite{PartII}, when characterizing the strong resilience of dynamical flow networks. As it will become clear in Section \ref{sec:proofthmuniquelimit}, the global convergence result mainly relies on Assumption \ref{ass:flowfunction} on monotonicity of the flow function, and Property (a) of Definition \ref{def:myopicpolicy} of locally responsive distributed routing policies, from which the dynamical flow network (\ref{dynsyst}) inherits a cooperative property. It is worth mentioning that we shall not use general results for cooperative dynamical systems \cite{Hirsch:82,Hirsch:85,Smith:95}, but rather exploit some other structural properties of (\ref{dynsyst}) which in fact allow us to prove stronger results. The additional property of the limit flow follows instead mainly from Property (b) of Definition \ref{def:myopicpolicy}.

\begin{figure}
\begin{center}
\subfigure[]{\includegraphics[width=7cm,height=5cm]{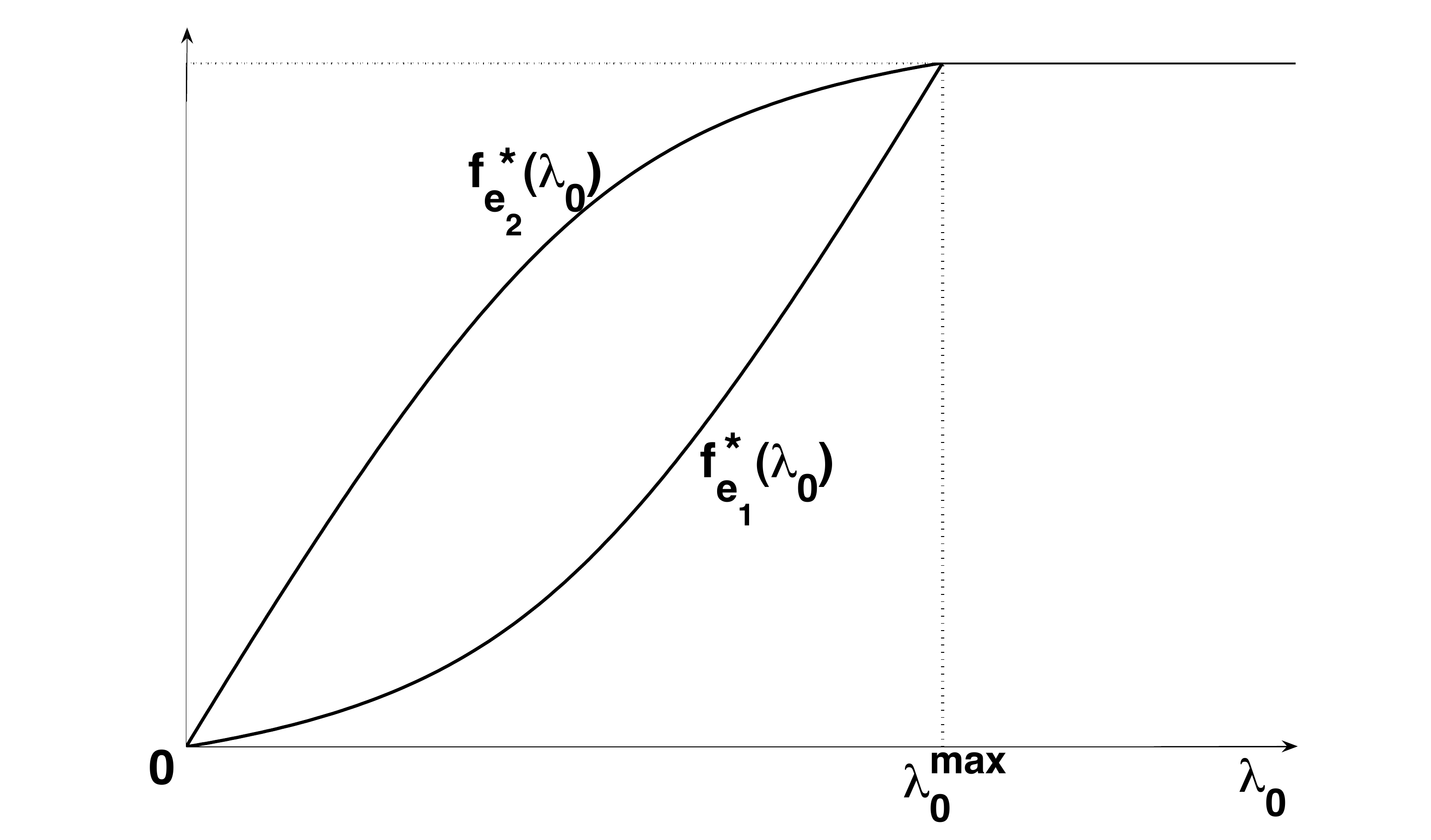}}\hspace{1cm}
\subfigure[]{\includegraphics[width=6cm,height=5cm]{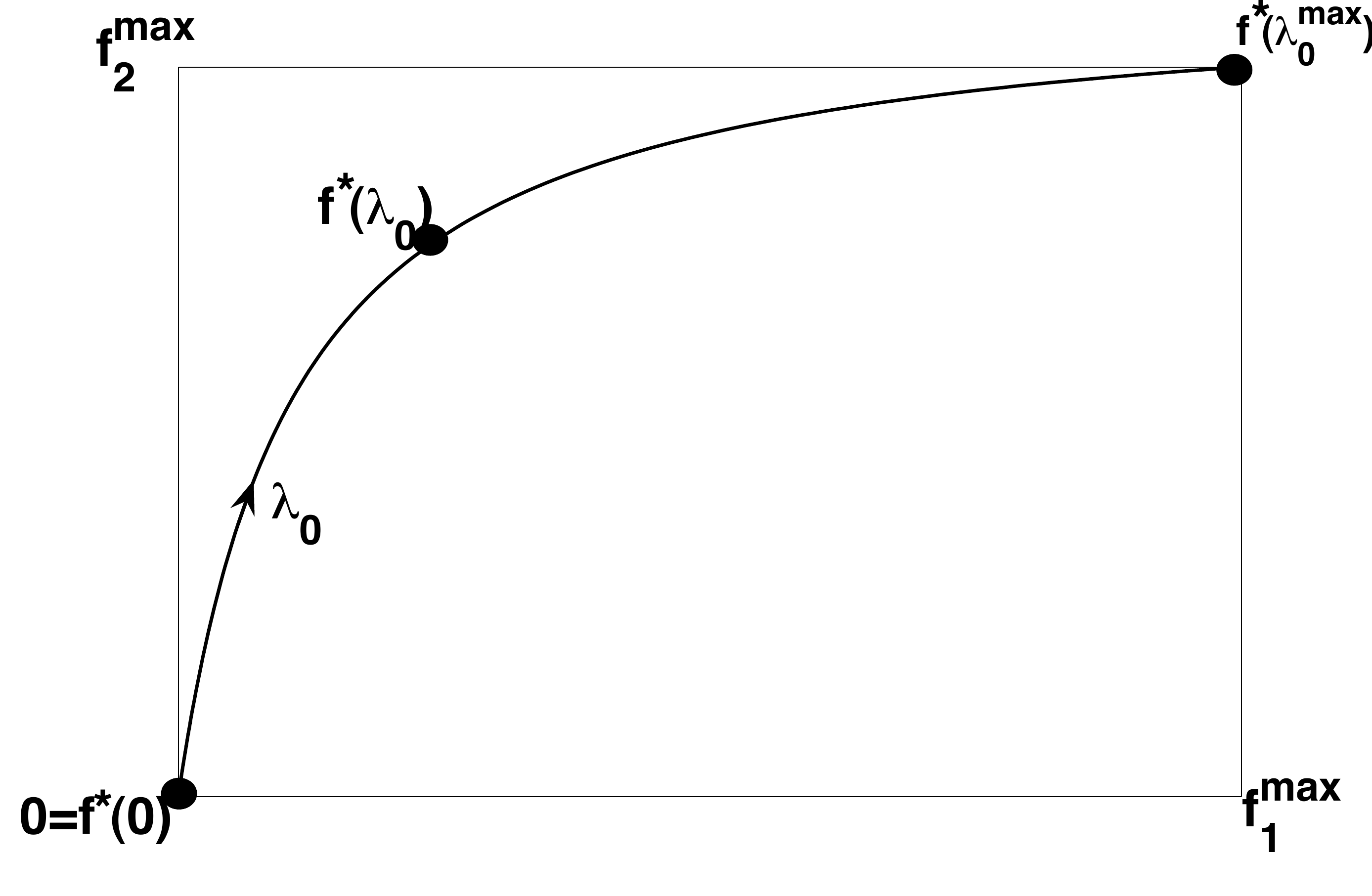}}
\end{center}
\caption{\label{fig:fstar}Dependence of the limit flow $f^*$ on the inflow $\lambda_0$ for the dynamical flow network of Example \ref{example:limitflow}. In (a), the two components of the limit flow, $f^*_{e_1}$ and $f^*_{e_2}$, are plotted as functions of the inflow $\lambda_0$. In (b), the curve of the limit flows is plotted in the $(f^*_{e_1},f^*_{e_2})$-plane. Observe as both components are increase from $0$ to $f_e^{\max}$, as $\lambda_0$ ranges between $0$ and $\lambda_0^{\max}$, while they remain constant at $f_e^{\max}$, as $\lambda_0$ varies above $\lambda_v^{\max}$. }
\end{figure}
\begin{figure}
\begin{center}
\subfigure[$\lambda_0=0$]{\includegraphics[width=5.4cm,height=4.8cm]{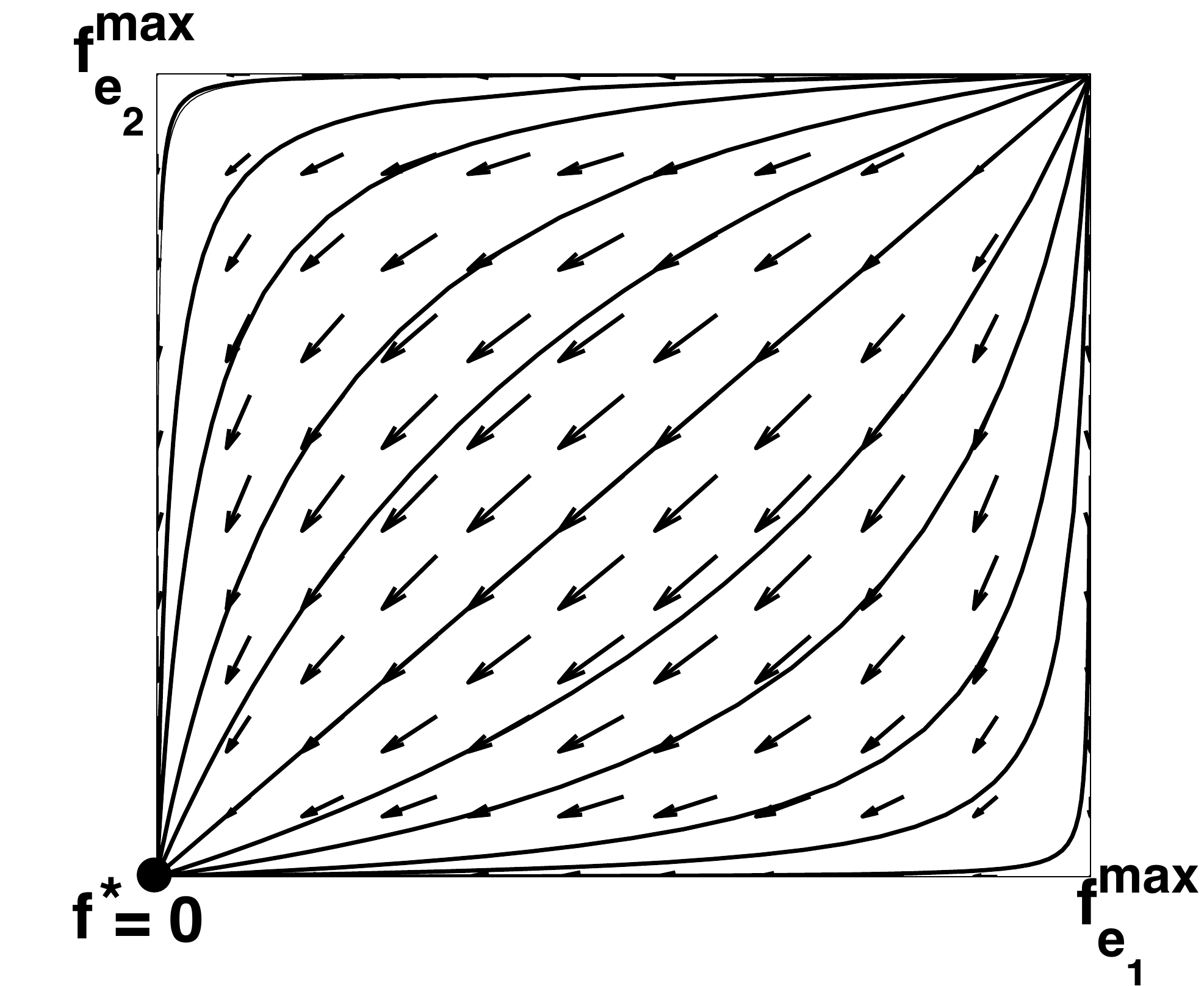}}
\subfigure[$\lambda_0=1$]{\includegraphics[width=5.4cm,height=4.8cm]{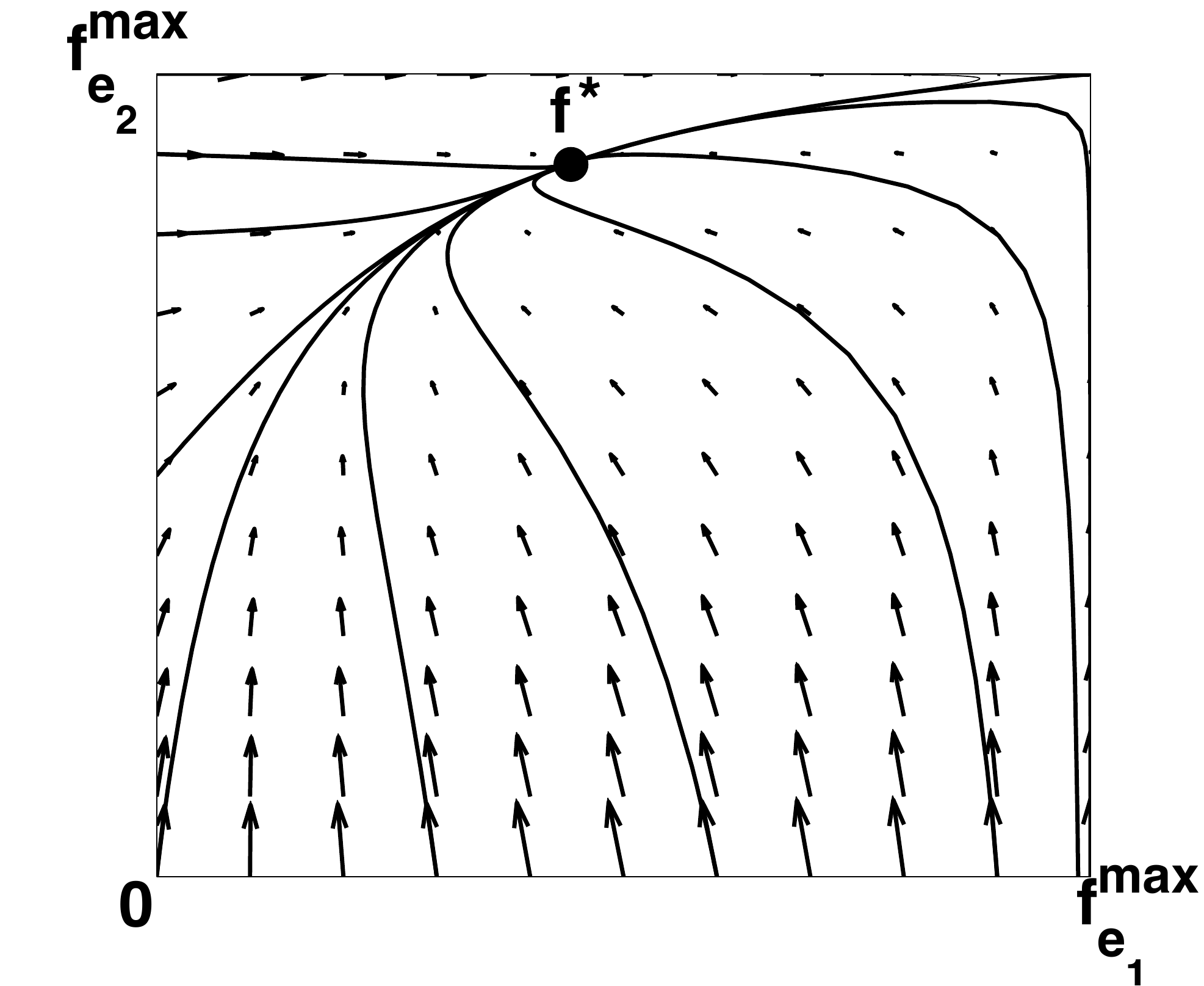}}
\subfigure[$\lambda_0=2$]{\includegraphics[width=5.4cm,height=4.8cm]{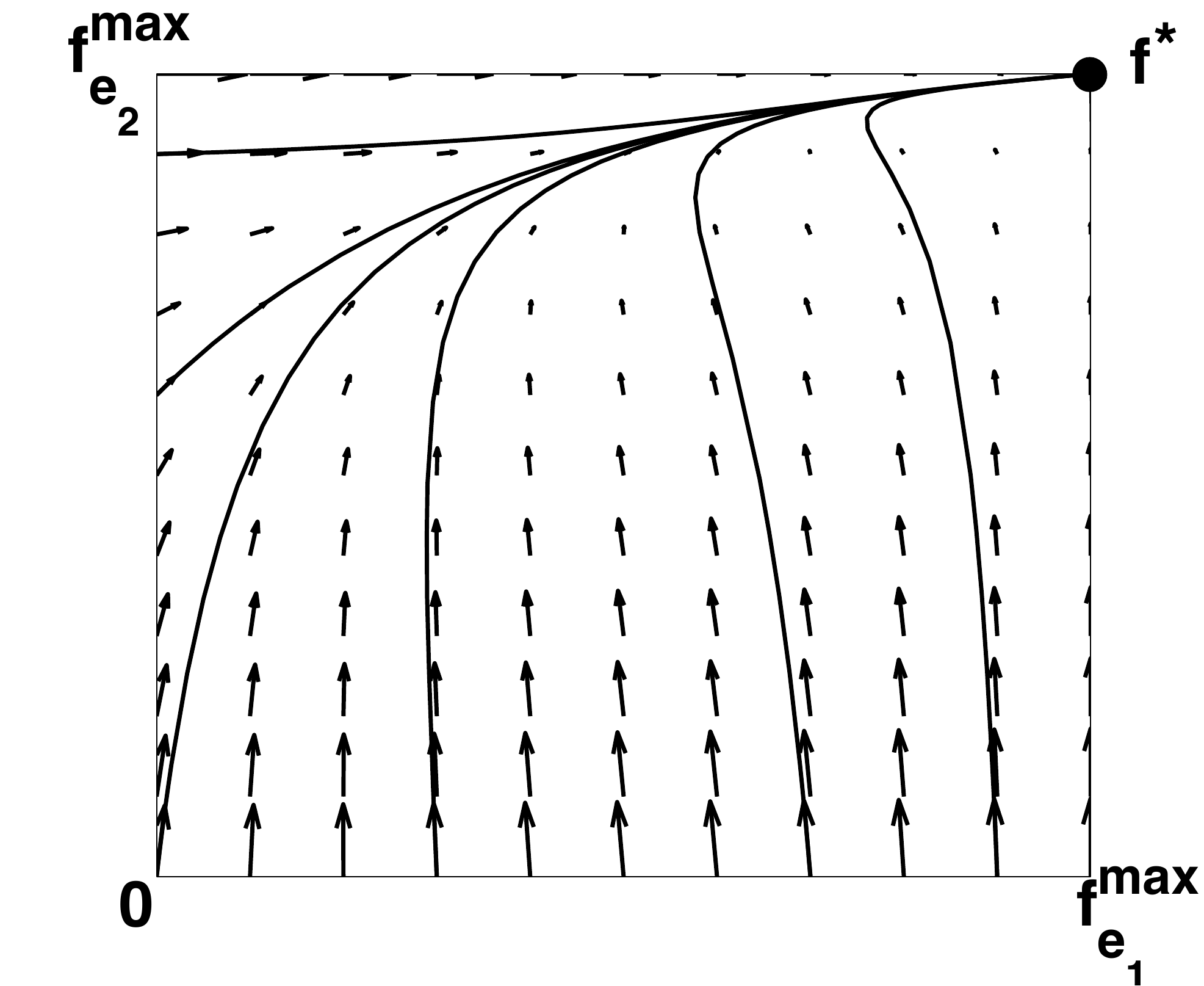}}
\end{center}
\caption{\label{fig:vectfield}Flow vector fields and flow trajectories for the dynamical flow network of Example \ref{example:limitflow}, for three values of the inflow. In the first two cases $\lambda_0<\lambda_0^{\max}$, and hence the limit flow $f^*$ is an equilibrium flow. In contrast, in the latter case, $\lambda_0\ge\lambda_0^{\max}$, and consequently $f^*$ is not an equilibrium flow and $f^*_{e_1}=f^{\max}_{e_1}$ and $f^{\max}_{e_2}=f^*_{e_2}$, as predicted by Theorem \ref{thm:uniquelimitflow}.}
\end{figure}

\begin{example}\label{example:limitflow}
Consider a simple topology containing just the origin and the destination node, i.e., with $\mc V=\{0,1\}$, and two parallel links $\mc E=\{e_1,e_2\}$. Assume that the flow functions on the two links are identical $\mu_{e_1}(\rho)=\mu_{e_2}(\rho)=3(1-e^{-\rho})/4$. Consider the routing policy 
$$G^0_{e_1}(\rho)=\frac{\frac35e^{-\rho_{e_1}}}{\frac35e^{-\rho_{e_1}}+6e^{-\rho_{e_2}}}\,,\qquad
G^0_{e_2}(\rho)=\frac{6e^{-\rho_{e_2}}}{\frac35e^{-\rho_{e_1}}+6e^{-\rho_{e_2}}}\,.$$
Then, the limit flow of the associated dynamical flow network can be explicitly computed for every constant inflow $\lambda_0\ge0$, and is given by 
$$f_1^*(\lambda_0)=\l\{\ba{lcl}
\l(12\lambda_0-11+\sqrt{(12\lambda_0-11)^2+28\lambda_0}\r)/{24}
&\text{ if }& 0\le\lambda_0<3/2\\
3/4&\text{ if }&\lambda_0\ge3/2\,,\ea\r.$$
$$f_2^*(\lambda_0)=\l\{\ba{lcl}
\l(12\lambda_0+11-\sqrt{(12\lambda_0-11)^2+28\lambda_0}\r)/{24}
&\text{ if }& 0\le\lambda_0<3/2\\
3/4&\text{ if }&\lambda_0\ge3/2\,.\ea\r.$$

Figure \ref{fig:fstar} shows the dependence of the limit flow $f^*$ on the inflow $\lambda_0$. The two components $f^*_{e_1}$, and $f^*_{e_2}$, increase from $0$ to $f_{e_1}^{\max}$, and, respectively, from $0$ to $f_{e_2}^{\max}$, as $\lambda_0$ ranges from $0$ to $\lambda_0^{\max}:=f^{\max}_{e_1}+f^{\max}_{e_2}$, while they remain constant as $\lambda_0$ varies above $\lambda_0^{\max}$. Figure \ref{fig:vectfield} reports the vector fields and flow trajectories associated to the dynamical flow network for three different values of the inflow, namely $\lambda_0=0$, $\lambda_0=1$, and $\lambda_0=2$. 
In the first two cases, $\lambda_0 < \lambda_0^{\max}$, and $f^*\in\mc F^*(\lambda_0)$ is an equilibrium flow; in the case (iii), $f^*\in\cl(\mc F^*(\lambda_0))\setminus\mc F^*(\lambda_0)$ is not an equilibrium flow. 
\end{example}\medskip

Our second main result, stated below, shows that locally responsive distributed routing policies are maximally robust, as the resilience of the induced dynamical flow network coincides with the min-cut capacity of the network. 

\begin{theorem}[Weak resilience for locally responsive distributed routing policies]
\label{maintheo-weakstability}
Let $\mc N$ be a flow network satisfying Assumptions \ref{ass:acyclicity} and \ref{ass:flowfunction}, $\lambda_0>0$ a constant inflow, and $\mc G$ a locally responsive distributed routing policy such that $G^v_e(\rho^v)>0$ for all $0\le v<n$, $e\in\mc E^+_v$, and $\rho^v\in\mc R_v$. Then, for every $f^{\circ}\in\mc F$, the associated dynamical flow network is partially transferring with respect to $f^{\circ}$ and has weak resilience
$$\gamma_0(f^{\circ},\mc G)=C (\mc N)\,.$$
\end{theorem}
\proof See Section \ref{sec:proof2}.\qed
\medskip

Theorem \ref{maintheo-weakstability}, combined with Proposition \ref{propUB}, shows that locally responsive distributed routing policies achieve the maximal weak resilience possible on a given flow network $\mc N$. A consequence of this result is that locality constraints on the feedback information available to routing policies do not reduce the achievable weak resilience. It is also worth observing that such maximal weak resilience coincides with min-cut capacity of the network, and is therefore independent of the initial flow $f^{\circ}$. This is in sharp contrast with the results on the strong resilience of dynamical flow networks presented in the companion paper \cite{PartII}. There, it is shown that the strong resilience depends on the initial flow, and local information constraints reduce the maximal strong resilience achievable on a given flow network.

\section{Proof of Theorem \ref{thm:uniquelimitflow}}\label{sec:proofthmuniquelimit}
Let $\mc N$ be a flow network satisfying Assumptions \ref{ass:acyclicity} and \ref{ass:flowfunction}, $\mc G$ a locally responsive distributed routing policy, and $\lambda_0\ge0$ a constant inflow. We shall prove that there exists a unique $f^*\in\cl(\mc F)$ such that the flow $f(t)$ associated to the solution of the dynamical flow network (\ref{dynsyst}) converges to $f^*$ as $t$ grows large, for every initial condition $\rho(0)\in\mc R$. Before proceeding, it is worth observing that, thanks to Property (a) of Definition \ref{def:myopicpolicy} of locally responsive distributed routing policies, Assumption \ref{ass:flowfunction} on the monotonicity of the flow functions, and the structure of the dynamical flow network (\ref{dynsyst}), one may rewrite (\ref{dynsyst}) as 
$$\frac{\de}{\de t}\rho_e=F_e(\rho)\,,\qquad \forall e\in\mc E\,,$$ where $F:\mc R\to\R^{\mc E}$ is differentiable and such that
$$\frac{\partial}{\partial\rho_e}F_e(\rho)\le0\,,\qquad \frac{\partial}{\partial\rho_j}F_e(\rho)\ge0\,,\qquad \forall e\ne j\in\mc E\,.$$
The above shows that, the dynamical flow network (\ref{dynsyst}) driven by a locally responsive distributed routing policy $\mc G$ is cooperative in the sense of Hirsch \cite{Hirsch:82,Hirsch:85}. Indeed, one may apply the standard theory of cooperative dynamical systems and monotone flows \cite{Hirsch:82,Hirsch:85,Smith:95} in order to prove some properties of the solution of (\ref{dynsyst}), e.g., convergence from almost every initial condition. However, we shall not rely on this general theory and rather use a direct approach based on a Lyapunov argument exploiting the particular structure of the dynamical system (\ref{dynsyst}), and leading 
us to stronger results, i.e., \emph{global} convergence to a \emph{unique} limit flow.
 
We shall proceed by proving a series of intermediate results some of which will prove useful also in the companion paper \cite{PartII}. First, given an arbitrary non-destination node $0\le v<n$, we shall focus on the input-output properties of the \emph{local system}
\be\label{localsys}\frac{\de}{\de t}\rho_e(t)=\lambda(t)G^v_e(\rho^v(t))-f_e(t)\,,\qquad f_e(t)=\mu_e(\rho_e(t))\,,\qquad\forall e\in\mc E^+_v\,,\ee
where $\lambda(t)$ is a nonnegative-real-valued, Lipschitz continuous input, and $f^v(t):=\{f_e(t):\,e\in\mc E^+_v\}$ is interpreted as the output. We shall first prove existence (and uniqueness) of a globally attractive limit flow for the system (\ref{localsys}) under constant input. We shall then extend this result to show the existence and attractivity of a local equilibrium point under time-varying, convergent local input. 
Finally, we shall exploit this local input-output property, and the assumption of acyclicity of the network topology in order to establish the main result.

The following is a simple technical result, which will prove useful in order to apply Property (a) of Definition \ref{def:myopicpolicy}.
\begin{lemma}
\label{lem:coop-ext}
Let $0\le v<n$ be a nondestination node, and $G^v:\mc R_v\to\mc S_v$ a continuously differentiable function satisfying Property (a) of Definition \ref{def:myopicpolicy}. Then, for any $\sigma,\varsigma\in\mc R_v$, 
\begin{equation}
\label{eq:coop-extended}
\sum\nolimits_{e \in \mc E_v^+} \sgn(\sigma_e- \varsigma_e) \left(G_e^v(\sigma)-G_e^v(\varsigma) \right) \, \le \, 0.
\end{equation}
\end{lemma}
\begin{proof}
Consider the sets $\mc K:=\{e\in\mc E_v^+:\,\sigma_e>\varsigma_e\}$, $\mc J:=\{e\in\mc E^+_v:\,\sigma_e\le\varsigma_e\}$, and $\mc L:=\{e\in\mc E_v^+:\,\sigma_e<\varsigma_e\}$. Define $G_{\mc K}(\zeta):=\sum_{k\in\mc K}G^v_k(\zeta)$, $G_{\mc L}(\zeta):=\sum_{l\in\mc L}G^v_l(\zeta)$, and $G_{\mc J}(\zeta):=\sum_{j\in\mc J}G_j^v(\zeta)$. We shall show that, for any $\sigma,\varsigma\in\mc R_v$,  
\begin{equation}
\label{eq:JK-ineq}
G_{\mc K}(\sigma) \le G_{\mc K}(\varsigma), \qquad G_{\mc L}(\sigma) \ge G_{\mc L}(\varsigma)\,.\end{equation}
Let $\xi\in\mc R_v$ be defined by $\xi_k=\sigma_k$ for all $k\in\mc K$, and $\xi_e=\varsigma_e$ for all $e\in\mc E^+_v\setminus\mc K$. We shall prove that $G_{\mc K}(\sigma)-G_{\mc K}(\varsigma)\le0$ by writing it as a path integral of $\nabla G_{\mc K}(\zeta)$ first along the segment $S_{\mc K}$ from $\varsigma$ to $\xi$, and then along the segment $S_{\mc L}$ from $\xi$ to $\sigma$. Proceeding in this way, one gets 
\be\label{GK-GK} G_{\mc K}(\sigma)-G_{\mc K}(\varsigma)=
\int_{S_{\mc K}}\nabla G_{\mc K}(\zeta)\cdot\de\zeta+\int_{S_{\mc L}}\nabla G_{\mc K}(\zeta)\cdot\de\zeta=
-\int_{S_{\mc K}}\nabla G_{\mc J}(\zeta)\cdot\de\zeta+\int_{S_{\mc L}}\nabla G_{\mc K}(\zeta)\cdot\de\zeta\,,\ee
where the second equality follows from the fact that $G_{\mc K}(\zeta)=1-G_{\mc J}(\zeta)$ since $G^v(\zeta)\in\mc S_v$. Now, Property (a) of Definition \ref{def:myopicpolicy} implies that $\partial G_{\mc K}(\zeta)/\partial\zeta_l\ge0$ for all $l\in\mc L$, and $\partial G_{\mc J}(\zeta)/\partial\zeta_k\ge0$ for all $k\in\mc K$. It follows that $\nabla G_{\mc J}(\zeta)\cdot\de\zeta\ge0$ along $S_{\mc K}$, and $\nabla G_{\mc K}(\zeta)\cdot\de\zeta\le0$ along $S_{\mc L}$. Substituting in (\ref{GK-GK}), one gets the first inequality in (\ref{eq:JK-ineq}). The second inequality in (\ref{eq:JK-ineq}) follows by similar arguments. Then, one has 
$$\sum\nolimits_{e \in \mc E_v^+} \sgn(\sigma_e- \varsigma_e) \left(G_e^v(\sigma)-G_e^v(\varsigma) \right)=G_{\mc K}(\sigma)-G_{\mc K}(\varsigma)+G_{\mc L}(\varsigma)-G_{\mc L}(\sigma)\le 0\,,$$
which proves the claim.\end{proof}
\medskip

We can now exploit Lemma~\ref{lem:coop-ext} in order to prove the following key result guaranteeing that the solution of the local dynamical system (\ref{localsys}) with constant input $\lambda(t)\equiv\lambda$ converges to a limit point which depends on the value of $\lambda$ but not on the initial condition. (Cf.~Example \ref{example:limitflow} and Figure \ref{fig:vectfield}.)

\begin{lemma}\textit{(Existence of a globally attractive limit flow for the local dynamical system under constant input)}
\label{lemmalocalexistence} 
Let $0\le v<n$ be a non-destination node, and $\lambda$ a nonnegative-real constant. Assume that $G^v:\mc R_v\to\mc S_v$ is continuously differentiable and satisfies Property (a) of Definition~\ref{def:myopicpolicy}. Then, there exists a unique $f^*(\lambda)\in\cl(\mc F_v)$ such that the solution of the dynamical system~\eqref{localsys} with constant input $\lambda(t)\equiv\lambda$ satisfies 
$$\lim_{t\to+\infty}f_e(t)=f_e^*(\lambda)\,,\qquad\forall e\in\mc E^+_v\,,$$
for every initial condition $\rho^v(0)\in\mc R_v$.
\end{lemma}
\begin{proof} 
Let us fix some $\lambda\in\R_+$. For every initial condition $\sigma\in\mc R_v$, and time $t\ge0$, let $\Phi^t(\sigma):=\rho^v(t)$ be the value of the solution of (\ref{localsys}) with constant input $\lambda(t)\equiv\lambda$ and initial condition $\rho(0)=\sigma$, at time $t \geq 0$. Also, let $\Psi^t(\sigma)\in\mc R_v$ be defined by $\Psi^t_e(\sigma)=\mu_e(\Phi^t_e(\sigma))$, for every $e\in\mc E^+_v$. Now, fix two initial conditions $\sigma,\varsigma\in\mc R_v$, and define 
$$\chi(t):=||\Phi^t(\sigma)-\Phi^t(\varsigma)||_1\,,\qquad\xi(t):=||\Psi^t(\sigma)-\Psi^t(\varsigma)||_1\,.$$
Since $\mu_e(\rho_e)$ is increasing by Assumption \ref{ass:flowfunction}, one has that 
\begin{equation}
\label{eq:same-signs}
\sgn(\Phi^t_e(\sigma)-\Phi^t_e(\varsigma))=\sgn(\Psi^t_e(\sigma)-\Psi^t_e(\varsigma))\,.
\end{equation}
On the other hand, using Lemma~\ref{lem:coop-ext}, one gets 
\begin{equation}
\label{eq:G-ineq}
\sum\nolimits_{e\in\mc E^+_v}\sgn(\Phi^t_e(\sigma)-\Phi^t_e(\varsigma))\l(G^v_e(\Phi^t(\sigma))-G_e^v(\Phi^t(\varsigma))\r) \, \le \, 0\,,\qquad \forall \, t\ge0\,.
\end{equation}
From \eqref{eq:same-signs} and \eqref{eq:G-ineq}, it follows that, for all $0 \le s\le t $, 
\be \label{eq:chidot}
\ba{rcl}
\chi(t)&=&||\Phi^t(\varsigma)-\Phi^t(\sigma)||_1 \\[5pt]
&=&\chi(s)+\ds\int_s^t\sum\nolimits_{e\in\mc E^+_v}\sgn(\Phi^u_e(\sigma)-\Phi^u_e(\varsigma))
\big(G^v_e(\Phi^u(\sigma))-G_e^v(\Phi^s(\varsigma))-\Psi_e^u(\sigma)+\Psi_e^u(\varsigma)\big)\de u\\[10pt]
&\le&
\chi(s)-\ds\int_s^t||\Psi^u(\sigma)-\Psi^u(\varsigma)||_1\de u
\\[10pt]
&=&\chi(s)-\ds\int_s^t\xi(u)\de u\,.\ea
\ee
Since $\chi(t)\ge0$, (\ref{eq:chidot}) implies that $\int_0^t\xi(u)\de u\le\chi(0)$ for all $t\ge0$. Since $\xi(u)\ge0$, it follows that the limit 
$$\int_0^{+\infty}\xi(u)\de u=\lim_{t\to+\infty}\int_0^{t}\xi(u)\de u\le\chi(0)$$
exists and is finite. Now, observe that $\xi(t)=||\mu^v(\Phi^t(\sigma))-\mu^v(\Phi^t(\varsigma))||_1$ is a uniformly continuous function of $t$ on $[0,+\infty)$, as it is the composition of the uniformly continuous functions $f^v\mapsto||f^v||_1$, $\rho^v\mapsto\mu^v(\rho^v)$ (whose uniform continuity follows from Assumption \ref{ass:flowfunction}), and $t\mapsto\Phi^t(\sigma)$ and $t\mapsto\Phi^t(\varsigma)$ (whose uniform continuity follows from being the solutions of the local dynamical system (\ref{localsys})). Hence, an application of Barbalat's lemma \cite[Lemma 4.2]{Khalil:96} implies that $\xi(t)$ converges to $0$, as $t$ grows large. That is,
\be\label{Phit->0}\lim_{t\to+\infty}||\Psi^t(\sigma)-\Psi^t(\varsigma)||_1=0\,,\qquad \forall\sigma,\varsigma\in\mc R_v\,.\ee
Now, for any $\sigma\in\mc R_v$, one can apply (\ref{Phit->0}) with $\varsigma:=\Phi^{\tau}(\sigma)$, and get that 
$$\lim_{t\to+\infty}||\Psi^t(\sigma)-\Psi^{t+\tau}(\sigma)||_1=
\lim_{t\to+\infty}||\Psi^t(\sigma)-\Psi^t(\Phi^{\tau}(\sigma))||_1=0\,,\qquad\forall\tau\ge0\,.$$
The above implies that, for any initial condition $\rho^v(0)=\sigma\in\mc R_v$, the flow $\Psi^t(\sigma)$ is Cauchy, and hence convergent to some $f^*(\lambda,\sigma)\in\cl({\mc F}_v)$. It follows from (\ref{Phit->0}) again, that 
$$||f^*(\lambda,\sigma)-f^*(\lambda,\varsigma)||_1=\lim_{t\to+\infty}||\Psi^t(\sigma)-\Psi^t(\varsigma)||_1=0\,,\qquad\forall\sigma,\varsigma\in\mc R_v\,,$$
which shows that the limit flow does not depend on the initial condition.
\end{proof}

Now, let us define $$\lambda_v^{\max}:=\sum\nolimits_{e\in\mc E^+_v}f_e^{\max}\,.$$ 
The following result characterizes the way the local limit flow $f^*(\lambda)$ depends on the local input $\lambda$. (Cf.~Example \ref{example:limitflow} and Figure \ref{fig:fstar}.)
\begin{lemma}[Dependence of the local limit flow on the input]\label{lemma:f*(lambda)}
Let $0\le v<n$ be a non-destination node, and $\lambda$ a nonnegative-real constant. Assume that $G^v:\mc R_v\to\mc S_v$ is continuously differentiable and satisfies Properties (a) and (b) of Definition~\ref{def:myopicpolicy}. Let $f^*(\lambda)\in\cl(\mc F_v)$ be the limit flow of the local system (\ref{localsys}) with constant input $\lambda(t) \equiv \lambda$, existence and uniqueness of which follow from Lemma \ref{lemmalocalexistence}. Then, 
\begin{description}
\item[(i)]  if $\lambda<\lambda_v^{\max}$, then 
$$f_e^*(\lambda)<f_e^{\max}\,,\qquad\lambda G^v_e({\mu}^{-1}(f^*(\lambda)))= f_e^*\,,\qquad\forall e\in\mc E^+_v\,;$$
\item[(ii)] if $\lambda\ge\lambda_v^{\max}$, then $f^*_e(\lambda)=f_e^{\max}$ for every $e\in\mc E^+_v$.
\end{description}
Moreover, $f^*(\lambda)$ is continuous as a map from $\R_+$ to $\cl(\mc F_v)$.\end{lemma} 
\begin{proof}
Define $\rho^*\in\mc R_v$ by 
$$\rho^*_e:=\l\{\ba{lcl}\mu_e^{-1}(f_e^*(\lambda))&\se&f_e^*(\lambda)<f_e^{\max}\\[3pt]+\infty&\se&f_e^*(\lambda)=f_e^{\max}\,.\ea\r.$$ 
Now, by contradiction, assume that there exists a  nonempty proper subset $\mc J\subset\mc E^+_v$ such that $\rho^*_j<+\infty$ for every $j\in\mc J$, and $\rho^*_e=+\infty$ for every $k\in\mc K:=\mc E^+_v\setminus\mc J$. Thanks to Property (b) of Definition \ref{def:myopicpolicy}, one would have that, for any initial condition $\rho(0)\in\mc R$, the solution of (\ref{localsys}) satisfies 
$$\lim_{t\to+\infty}\sum_{k\in\mc K}\lambda G^v_k(\rho^v(t))-f_k(t)=-\sum_{k\in\mc K}f^{\max}_k<0\,,$$ 
so that there would exist some $\tau\ge0$ such that $$\sum_{k \in \mc K}(\lambda G^v_k(\rho^v(t))-f_k(t))\le0\,,\qquad \forall t\ge\tau\,.$$ Hence, 
$$
\sum_{k\in\mc K} \rho_k(t)=
\sum_{k\in\mc K} \rho_k(\tau)+ \int_\tau^t
\sum_{k\in\mc K} \l(\lambda G^v_k(\rho^v(s))-f_k(s))\r)\de s\le
\sum_{k\in\mc K}\rho_k(\tau)<+\infty\,, \qquad \forall t \geq \tau,
$$
which would contradict the assumption that $\rho^*_k=+\infty$ for every $k\in\mc K$. Therefore, either $\rho_e^*$ is finite for every $e\in\mc E^+_v$, or $\rho_e^*$ is infinite for every $e\in\mc E^+_v$.

In order to distinguish between the two cases, let $$\zeta(t):=\sum_{e\in\mc E^+_v}\rho_e(t)\,,\qquad\vartheta(t):=\sum_{e\in\mc E^+_v}f_e(t)\,.$$ Observe that, for all $t\ge\tau\ge0$, 
\be\label{chitxit}\zeta(t)=\zeta(\tau)+\int_{\tau}^t\l(\lambda-\vartheta(s)\r)\de s\,.\ee
First, consider the case when $\lambda<\lambda_v^{\max}$, and assume by contradiction that $\rho_e^*=+\infty$, and hence 
$f_e^*=f^{\max}_e$, for every $e\in\mc E^+_v$. This would imply that 
$$\lim_{t\to\infty}\vartheta(t)=\lambda_v^{\max}>\lambda\,,$$ 
so that there would exist some $\tau\ge0$ such that $\lambda-\vartheta(t)\le 0$ for every $t\ge\tau$, and hence (\ref{chitxit}) would imply that $\zeta(t)\le\zeta(\tau)<+\infty$ for all $t\ge\tau$, thus contradicting the assumption that $\rho_e(t)$ converges to $\rho_e^*=+\infty$ as $t$ grows large. Hence, necessarily $\rho^*\in\mc R_v$, and $f^*(\lambda)\in\mc F_v$. Therefore, being a finite limit point of the autonomous dynamical system (\ref{localsys}) with continuous right hand side, $\rho^*$ is necessarily an equilibrium, and so $f^*(\lambda)$ is an equilibrium flow for the local dynamical system (\ref{localsys}). 

On the other hand, when $\lambda\ge\lambda_v^{\max}$, (\ref{chitxit}) shows that $\zeta(t)$ is non-decreasing, hence convergent to some $\zeta(\infty)\in[0,+\infty]$ at $t$ grows large. Assume, by contradiction, that $\zeta(\infty)$ is finite. Then, passing to the limit of large $t$ in (\ref{chitxit}), one would get $$\int_{\tau}^{+\infty}(\lambda-\vartheta(s))\de s=\zeta(\infty)-\zeta(\tau)\le\zeta(\infty)<+\infty\,.$$
This, and the fact that $\vartheta(t)<\lambda_v^{\max}\le\lambda$ for all $t\ge0$, would imply that 
\be\label{xitimpossible}\lim_{t\to+\infty}\vartheta(t)=\lambda\,.\ee Since $f_e(t)<f_e^{\max}$, (\ref{xitimpossible}) is impossible if $\lambda>\lambda_v^{\max}$. On the other hand, if $\lambda=\lambda_v^{\max}$, then (\ref{xitimpossible}) implies that, for every $e\in\mc E^+_v$, $f_e(t)$ converges to $f_e^{\max}$, and hence $\rho_e(t)$ grows unbounded as $t$ grows large, so that $\zeta(\infty)$ would be infinite. Hence, if $\lambda\ge\lambda_v^{\max}$, then necessarily $\zeta(\infty)$ is infinite, and thanks to the previous arguments this implies that $\rho_e^*=+\infty$, and hence $f_e^*(\lambda) = f_e^{\max}$ for all $\sigma \in \mc R_v$,  $e\in\mc E^+_v$. 

Finally, it remains to prove continuity of $f^*(\lambda)$ as a function of $\lambda$. For this, consider the function $H:(0,+\infty)^{\mc E^+_v}\times(0,\lambda_v^{\max})\to\R^{\mc E^+_v}$ defined by 
$$H_e(\rho^v,\lambda):=\lambda G^v_e(\rho^v)-\mu_e(\rho_e)\,,\qquad \forall e\in\mc E^+_v\,.$$ 
Clearly, $H$ is differentiable and such that 
\begin{equation}
\label{eq:jacobian}
\frac{\partial}{\partial_{\rho_e}}H_e(\rho^v,\lambda)=\lambda\frac{\partial}{\partial_{\rho_e}}G^v_e(\rho^v)-\mu_e'(\rho_e)=-\sum_{j\ne e}\lambda\frac{\partial}{\partial_{\rho_e}}G^v_j(\rho^v)-\mu_e'(\rho_e)<-\sum_{j\ne e}\frac{\partial}{\partial_{\rho_e}}H_j(\rho^v,\lambda)\,,
\end{equation}
where the inequality follows from the strict monotonicity of the flow function (see Assumption \ref{ass:flowfunction}). Property (a) in Definition~\ref{def:myopicpolicy} implies that $\partial H_j(\rho^v,\lambda)/\partial \rho_e \geq 0$ for all $j \neq e \in \mc E_v^+$. Hence, from (\ref{eq:jacobian}), we also have that $\partial H_e(\rho^v,\lambda)/\partial \rho_e < 0$ for all $e \in \mc E_v^+$.
Therefore, for all $\rho^v\in(0,+\infty)^{\mc E^+_v}$, and $\lambda\in(0,\lambda_v^{\max})$, the Jacobian matrix $\nabla_{\rho^v} H(\rho^v,\lambda)$ is strictly diagonally dominant, and hence invertible by a standard application of the Gershgorin Circle Theorem, e.g., see Theorem 6.1.10 in \cite{Horn.Johnson:90}. It then follows from the implicit function theorem that $\rho^*(\lambda)$, which is the unique zero of $H(\,\cdot\,,\lambda)$, is continuous on the interval $(0,\lambda_v^{\max})$. Hence, also $f^*(\lambda)=\mu(\rho^*(\lambda))$ is continuous on $(0,\lambda_v^{\max})$, since it is the composition of two continuous functions. Moreover, since 
$$\sum_{e\in\mc E^+_v}f_e^*(\lambda)=\lambda\,,\qquad 0\le f_e^*(\lambda)\le f^{\max}_e\,,\qquad \forall e\in\mc E^+_v\,,\qquad \forall\lambda\in(0,\lambda_v^{\max})\,,$$
one gets that 
$$\lim_{\lambda\downarrow0}f_e^*(\lambda)=0\,,\qquad\lim_{\lambda\uparrow\lambda_v^{\max}} f_e^*(\lambda)=f_e^{\max}\,,$$ for all $e\in\mc E^+_v$. Now, one has that $\sum_{e\in\mc E^+_v}f_e^*(0)=0$, so that $$0=f_e^*(0)=\lim_{\lambda\downarrow0}f_e^*(\lambda)\,,\qquad\forall e\in\mc E^+_v\,.$$ Moreover, as previously shown, $$f_e^*(\lambda)=f_e^{\max}=\lim_{\lambda\uparrow\lambda_v^{\max}} f_e^*(\lambda)\,,\qquad \forall \lambda\ge\lambda_v^{\max}\,.$$ This completes the proof of continuity of $f^*(\lambda)$ on $[0,+\infty)$.
\end{proof}\medskip

While Lemma~\ref{lemmalocalexistence} ensures existence of a unique limit point for the local system (\ref{localsys}) with constant input $\lambda(t)\equiv\lambda$, the following lemma establishes a  monotonicity property with respect to a time-varying input $\lambda(t)$. 
\begin{lemma}[Monotonicity of the local system]\label{lemma:monotone}
Let $0\le v<n$ be a nondestination node, $G^v:\mc R_v\to\mc S_v$ a continuously differentiable map, satisfying Properties (a) and (b) of Definition~\ref{def:myopicpolicy}, and $\lambda^-(t)$, and $\lambda^+(t)$ be two nonnegative-real valued Lipschitz-continuous functions such that $\lambda^-(t)\le\lambda^+(t)$ for all $t \geq 0$. Let $\rho^-(t)$ and $\rho^+(t)$ be the solutions of the local dynamical system (\ref{localsys}) corresponding to the inputs $\lambda^-(t)$, and $\lambda^+(t)$, respectively, with the same initial condition $\rho^-(0)=\rho^+(0)$. Then 
\be\label{eq:monotonicity}
\rho^-_e(t)\le\rho^+_e(t)\,,\qquad\forall e\in\mc E^+_v\,,\ \forall t\ge0\,.
\ee
\end{lemma}
\begin{proof} 
For $e\in\mc E^+_v$, define $\tau_e:=\inf\{t\ge0:\,\rho^+_e(t)>\rho^-_e(t)\}$, and let $\tau:=\min\{\tau_e:\,e\in\mc E^+_v\}$. Assume by contradiction that $\rho^-_e(t)> \rho^+_e(t)$ for some $t\ge0$, and $e\in\mc E^+_v$. Then, $\tau<+\infty$, and $\mc I:=\argmin\{\tau_e:\,e\in\mc E^+_v\}$ is a well defined nonempty subset of $\mc E^+_v$. Moreover, by continuity, one has that there exists some $\eps>0$ such that, $\rho^-_i(\tau)=\rho^+_i(\tau)$, $\rho^-_i(t)>\rho^+_i(t)$, and $\rho^-_j(t)<\rho^+_j(t)$ for all $i\in\mc I$, $j\in\mc J$, and $t\in(\tau,\tau+\eps)$, where $\mc J:=\mc E^+_v\setminus\mc I$. Using Lemma \ref{lem:coop-ext}, one gets, for every $t\in(\tau,\tau+\eps)$,
$$\ba{rcl}
0&\ge&
\frac12\sum_{e}\sgn(\rho^-_e(t)-\rho^+_e(t))\l(G^v_e(\rho^-(t))-G^v_e(\rho^+(t))\r)\\[7pt]
&=&\frac12\l(\sum_{i}G^v_i(\rho^-(t))-\sum_{i}G^v_i(\rho^+(t))
-\sum_{j}G^v_j(\rho^-(t))+\sum_{j}G^v_j(\rho^+(t))\r)\\[7pt]
&=&\sum_{i}G^v_i(\rho^-(t))-\sum_{i}G^v_i(\rho^+(t))\,,\ea$$
where the summation indices $e$, $i$, and $j$ run over $\mc E^+_v$, $\mc I$, and $\mc J$, respectively. 
On the other hand, Assumption \ref{ass:flowfunction} implies that $\mu_i(\rho^-_i(t))\ge\mu_i(\rho^+_i(t))$ for all $i\in\mc I$, $t\in[\tau,\tau+\eps)$. Now, let $\chi(t):=\sum_{i\in\mc I} \left( \rho^-_i(t)-\rho^+_i(t) \right)\,.$
Then, for every $t\in(\tau,\tau+\eps)$, one has 
$$
\ba{rclcl}0&<&
\chi(t)-\chi(\tau)\\[7pt]&=&\ds
\int_\tau^t\lambda^-(s)\sum\nolimits_{i \in \mc I}\l(G^v_i(\rho^-(s))-G^v_i(\rho^-(s))\r)\de s\\[7pt]
&&\ds
-\int_\tau^t(\lambda^+(s)-\lambda^-(s))\sum\nolimits_{i \in \mc I}G^v_i(\rho^+(s))\de s
-\int_\tau^t\sum\nolimits_{i \in \mc I}\l(\mu_i(\rho^-_i(s))-\mu_i(\rho^+_i(s))\r)\de s\\&\le&0\,,
\ea
$$
which is a contradiction. Then, necessarily (\ref{eq:monotonicity}) has to hold true. 
\end{proof}\medskip

The following lemma establishes that the output of the local system (\ref{localsys}) is convergent, provided that the input is convergent. 
\begin{lemma}[Attractivity of the local dynamical system]
\label{lemma:attractivity}
Let $0\le v<n$ be a nondestination node, $G^v:\mc R_v\to\mc S_v$ a continuously differentiable map, satisfying Properties (a) and (b) of Definition~\ref{def:myopicpolicy}, and $\lambda(t)$ a nonnegative-real-valued Lipschitz continuous function such that 
\be\label{limlambda}
\lim_{t\to+\infty}\lambda(t)= \lambda\,.\ee
Then, for every initial condition $\rho(0)\in\mc R$, the solution of the local dynamical system (\ref{localsys}) satisfies
\be
\label{eq:exp-conv}\lim_{t\to+\infty}f_e(t)=f_e^{*}(\lambda)\,,\qquad \forall e \in \mc E_v^+\,,\ee
where $f^{*}(\lambda)$ is as defined in Lemma \ref{lemmalocalexistence}.
\end{lemma}
\begin{proof}
Fix some $\eps>0$, and let $\tau\ge0$ be such that $|\lambda(t)-\lambda|\le\eps$ for all $t\ge\tau$. For $t\ge\tau$, let $f^-(t)$ and $f^+(t)$ be the flow associated to the solutions of the local dynamical system (\ref{localsys}) with initial condition $\rho^-(\tau)=\rho^+(\tau)=\rho^v(\tau)$, and constant inputs $\lambda^-(t)\equiv\lambda^-:=\max\{\lambda-\eps,0\}$, and $\lambda^+(t)\equiv\lambda+\eps$, respectively. From Lemma \ref{lemma:monotone}, one gets that
\be\label{sandwitch}f^-_e(t)\le f_e(t)\le f^+_e(t)\,,\qquad \forall t\ge\tau\,,\qquad\forall e\in\mc E^+_v\,.\ee
On the other hand, Lemma \ref{lemmalocalexistence} implies that $f^-(t)$ converges to $f^*(\lambda^-)$, and $f^+(t)$ converges to $f^*(\lambda^+)$, as $t$ grows large. Hence, passing to the limit of large $t$ in (\ref{sandwitch}) yields 
$$f_e^*(\lambda^-)\le\liminf_{t\to+\infty}f_e(t)\le\limsup_{t\to+\infty}f_e(t)\le f^*_e(\lambda+\eps)\,,\qquad\forall e\in\mc E^+_v\,.$$
Form the arbitrariness of $\eps >0$, and the continuity of $f^*(\lambda)$ as a function of $\lambda$ by Lemma~\ref{lemma:f*(lambda)}, it follows that $f(t)$ converges to $f^*(\lambda)$, as $t$ grows large, which proves the claim. 
\end{proof}\medskip

We are now ready to prove Theorem \ref{thm:uniquelimitflow} by showing that, for any initial condition $\rho(0)\in\mc R$, the solution of the dynamical flow network (\ref{dynsyst}) satisfies 
\be\label{tildefe*}\lim_{t\to+\infty}f_e(t)=f_e^*\,,\ee  for all $e\in\mc E$.
We shall prove this by showing via induction on $v=0,1,\ldots,n-1$ that, for all $e\in\mc E^+_v$, there exists $f_e^*\in[0,f_e^{\max}]$ such that (\ref{tildefe*}) holds true. First, observe that, thanks to Lemma \ref{lemmalocalexistence}, this statement is true for $v=0$, since the inflow at the origin is constant. Now, assume that the statement is true for all $0\le v<w$, where $w\in\{1,\ldots,n-2\}$ is some intermediate node. Then, since $\mc E^-_w\subseteq\cup_{v=0}^{w-1}\mc E^+_v$, one has that $$\lim_{t\to+\infty}\lambda^-_w(t)=\lim_{t\to+\infty}\sum\nolimits_{e\in\mc E^-_w}f_e(t)=\sum\nolimits_{e\in\mc E^-_w}f_e^*=\lambda^*_w\,.$$ Then, Lemma \ref{lemma:attractivity} implies that, for all $e\in\mc E^+_w$, (\ref{tildefe*}) holds true with $f_e^*=f_e^*(\lambda^*_w)$, thus proving the statement for $v=w$. This proves the existence of a globally attractive limit flow $f^*$. The proof of Theorem \ref{thm:uniquelimitflow} is completed by Lemma~\ref{lemma:f*(lambda)}.
\medskip

 \section{Proof of Theorem~\ref{maintheo-weakstability}}
\label{sec:proof2} 
This section is devoted to the proof of Theorem~\ref{maintheo-weakstability} on the weak resilience of dynamical flow networks with locally responsive distributed routing policies $\mc G$. 


To start with, let us recall that in this case Theorem \ref{thm:uniquelimitflow} implies the existence of a globally attractive limit flow $\tilde f^*\in\cl(\mc F)$ for the perturbed dynamical flow network associated to any admissible perturbation $\tilde{\mc N}$. Define $\tilde\lambda^*_0=\lambda_0$, and $\tilde\lambda^*_v=\sum_{e\in\mc E^-_v}\tilde f^*_e$, for $0<v\le n$.

 \begin{lemma}\label{lemma:enoughflow}
Consider a dynamical flow network $\mc N$ satisfying Assumptions \ref{ass:acyclicity} and \ref{ass:flowfunction}, with locally responsive distributed routing policy $\mc G$ such that $G^v_e(\rho^v)>0$ for all $0\le v<n$, $e\in\mc E^+_v$, and $\rho^v\in\mc R_v$. Then, for every $\theta\ge1$, there exists $\beta_{\theta} \in (0,1)$ such that, if $\tilde{\mc N}$ is an admissible perturbation of $\mc N$ with stretching coefficient less than or equal to $\theta$, and $\tilde f^* \in\cl(\tilde{\mc F})$ is the limit flow vector of the corresponding perturbed dynamical flow network (\ref{pertdynsyst}), then 
$$\tilde f_e^*\ge\beta_{\theta}\tilde\lambda_v^*\,,$$ for every non-destination node $0\le v<n$, and every link $e\in\mc E^+_v$ for which  $\tilde f_e^*\le\tilde f_e^{\max}/2$.\end{lemma}
\begin{proof}
First, observe that the claim is trivially true if $\tilde f_e^*>\tilde f_e^{\max}/2$ for all $e\in\mc E$. Therefore, let us assume that there exists some link $e\in\mc E$ for which  $\tilde f_e^*\le\tilde f_e^{\max}/2$. Define $\rho^{\theta}\in\mc R_v$ by $\rho^\theta_j=0$ for all $j\in\mc E^+_v$, $j\ne e$, and $\rho^{\theta}_e=\theta\rhomedian_e$, where recall that $\rhomedian_e$ is the median density of the flow function $\mu_e$. Since the stretching coefficient of $\tilde{\mc N}$ is less than or equal to $\theta$, one has that the median densities of the perturbed and the unperturbed flow functions satisfy $\tilderhomedian_e\le\theta\rhomedian_e$. This and the fact that $\tilde f_e^*\le\tilde f_e^{\max}/2$ imply that $\tilde\rho_e^*\le\tilderhomedian_e\le\rho^{\theta}_e$, while clearly $\tilde\rho_j^*\ge0=\rho^{\theta}_j$ for all $j\in\mc E^+_v$, $j\ne e$. Now, let $\beta_{\theta}:=G^v_e(\rho^{\theta})$, and observe that, thanks to the assumption on the strict positivity of $G^v_e(\rho^v)$, one has $\beta_{\theta}>0$. Then, from Lemma \ref{lem:coop-ext} one gets that 
\be\label{Gveineq}G^v_e(\tilde\rho^*)=\frac12\l(G^v_e(\tilde\rho^*)+1-\sum\nolimits_{j\ne e}G^v_j(\tilde\rho^*)\r)\ge\frac12\l(G^v_e(\rho^{\theta})+1-\sum\nolimits_{j\ne e}G^v_j(\rho^{\theta})\r)= G^v_e(\rho^{\theta})=\beta_{\theta}\,.\ee
On the other hand, since $\tilde f_e^*\le\tilde f_e^{\max}/2<\tilde f^{\max}_e$, Lemma \ref{lemmalocalexistence} implies that necessarily $\tilde\lambda_v^*G^v_e(\tilde\rho^*)=\tilde f_e^*$. The claim now follows by combining this and (\ref{Gveineq}). 
\end{proof}

As a consequence of Lemma \ref{lemma:enoughflow}, we now prove the following result showing that the dynamical flow network is partially transferring and providing a lower bound on its weak resilience:  
\begin{lemma} \label{lemma:LBgammathetaalpha}
Let $\mc N$ be a flow network satisfying Assumptions \ref{ass:acyclicity} and \ref{ass:flowfunction}, $\lambda_0\ge0$ a constant inflow, and $\mc G$ a locally responsive distributed routing policy such that $G^v_e(\rho^v)>0$ for all $0\le v<n$, $e\in\mc E^+_v$, and $\rho^v\in\mc R_v$. Then, the associated dynamical flow network is partially transferring, and, for every $\theta\ge1$, and $\alpha\in(0,\beta_{\theta}^{n}]$, its resilience satisfies
$$\gamma_{\alpha,\theta}(\floweq,\mc G)\ge C(\mc N)-2|\mc E|\lambda_0\beta_{\theta}^{1-n}\alpha\,,$$
where $\beta_\theta\in(0,1)$ is as in Lemma~\ref{lemma:enoughflow}.
\end{lemma}
\begin{proof} 
Consider an arbitrary admissible perturbation $\tilde {\mc N}$ of magnitude \be\label{deltahypothesis}\delta\le C(\mc N)-2|\mc E|\lambda_0\beta_{\theta}^{1-n}\alpha\,,\ee
and stretching coefficient less than or equal to $\theta$. 
We shall iteratively select a sequence of nodes $0=:v_0,v_1,\ldots,v_k:=n$ such that, for every $1\le j\le k$, 
\be\label{existseinDv}\exists i\in\{0,\ldots,j-1\}\qquad\text{such that}\qquad (v_i,v_j)\in\mc E\,,\quad
\tilde f^*_{(v_i,v_j)}\ge\lambda_0\alpha\beta_{\theta}^{j-n}\,.\ee
Since ${v_k}=n$, and $\beta_{\theta}^{k-n}\ge1$, the above with $j=k\le n$ will immediately imply that 
\be\label{alphatranfer}\lim_{t\to+\infty}\tilde\lambda_n(t)=\tilde\lambda^*_n=\sum\nolimits_{e\in\mc E^-_n}\tilde f_e^*\ge\alpha\lambda_0\beta_{\theta}^{k-n}\ge\alpha\lambda_0\,,\ee
so that the perturbed dynamical flow network is $\alpha$-transferring. For $0<\alpha\le \beta_{\theta}^{n-1}/(2|\mc E|\lambda_0)$, one could chose a trivial perturbation $\tilde{\mc N}=\mc N$ so that (\ref{alphatranfer}) would imply the partial transferring property of the original dynamical flow network. Moreover, the rest of the claim will then readily follow from the arbitrariness of the considered admissible perturbation. 

First, let us consider the case $j=1$. Assume by contradiction that $\tilde f^*_e<\lambda_0\alpha\beta_{\theta}^{1-n}$, for every link $e\in\mc E^+_0$. Since $\alpha\le\beta_{\theta}^n$, this would imply that $\tilde f^*_e<\beta_{\theta}\lambda_0$ and hence, by Lemma  \ref{lemma:enoughflow}, that $\tilde f^{\max}_e\le2\tilde f_e^*$ for all $e\in\mc E^+_0$, so that 
$$\sum\nolimits_{e\in\mc E^+_0}\tilde f^{\max}_e\le2\sum\nolimits_{e\in\mc E^+_0}\tilde f^{*}_e<2 \alpha |\mc E_0^+|\beta_{\theta}^{1-n}\lambda_0\le2 \alpha |\mc E|\beta_{\theta}^{1-n}\lambda_0\,.$$
Combining the above with the inequality $C(\mc N)\le\sum_{e\in\mc E^+_0} f_e^{\max}$, one would get 
$$\delta\ge\sum\nolimits_{e\in\mc E^+_0} \left(f^{\max}_e-\tilde f^{\max}_e\right)> C(\mc N)-2 \alpha |\mc E|\beta_{\theta}^{1-n}\lambda_0\,,$$ thus contradicting the assumption (\ref{deltahypothesis}). Hence, necessarily there exists $e\in\mc E^+_0$ such that $\tilde f^*_e\ge\lambda_0\alpha\beta_{\theta}^{1-n}$, and choosing $v_1$ to be the unique node in $\mc V$ such that $e\in\mc E^-_{v_1}$, one sees that (\ref{existseinDv}) holds true with $j=1$.

Now, fix some $1< j^*\le k$, and assume that (\ref{existseinDv}) holds true for every $1\le j<j^*$. Then, by choosing $i$ as in (\ref{existseinDv}), one gets that 
\be\label{ineq1}\tilde\lambda_{v_j}^*=\sum\nolimits_{e\in\mc E^+_{v_j}}\tilde f^*_e\ge\tilde f^*_{(v_i,v_{j})}\ge\lambda_0\alpha\beta_{\theta}^{j-n}\ge\lambda_0\alpha\beta_{\theta}^{j^*-1-n}\,,\qquad \forall 1\le j<j^*\,.\ee
Moreover, 
\be\label{ineq2}\tilde\lambda_{v_0}^*=\lambda_0>\lambda_0\alpha\beta_{\theta}^{-n}\ge\lambda_0\alpha\beta_{\theta}^{j^*-1-n}\,.\ee
Let $\mc U:=\{v_0,v_1,\ldots,v_{j^*-1}\}$ and $\mc E_{\mc U}^+\subseteq\mc E$ be the set of links with tail node in $\mc U$ and head node in $\mc V\setminus\mc U$. Assume by contradiction that 
$$\tilde f^*_e<\lambda_0\alpha\beta_{\theta}^{j^*-n}\,,\qquad \forall e\in\mc E^+_{\mc U}\,.$$ Thanks to (\ref{ineq1}) and (\ref{ineq2}), this would imply that, $\tilde f^*_e<\beta_{\theta}\tilde\lambda_j^*$, for every $0\le j<j^*$ and $e\in\mc E^+_{v_j}\cap\mc E^+_{\mc U}$. Then, since $\mc E^+_{\mc U}=\cup_{j=0}^{j^*-1}(\mc E_{v_j}^+\cap\mc E^+_{\mc U})$, Lemma \ref{lemma:enoughflow} would imply that 
$\tilde f^{\max}_e\le 2\tilde f^*_e$,for every $e\in\mc E^+_{\mc U}\,.$ 
This would yield
$$\sum\nolimits_{e\in\mc E_{\mc U}^+}\tilde f^{\max}_e\le\sum\nolimits_{e\in\mc E_{\mc U}^+}2\tilde f^{*}_e<2\sum\nolimits_{e\in\mc E_{\mc U}^+}\lambda_0\alpha\beta_{\theta}^{j^*-n}\le2|\mc E|\lambda_0\alpha\beta_{\theta}^{1-n}\,.$$
From the above and the inequality $C(\mc N)\le\sum_{e\in\mc E^+_{\mc U}} f_e^{\max}$, one would get 
$$\delta\ge\sum\nolimits_{e\in\mc E^+_{\mc U}} \left( f^{\max}_e-\tilde f^{\max}_e \right)> C(\mc N)-2 \alpha |\mc E|\beta_{\theta}^{1-n}\lambda_0\,,$$ thus contradicting the assumption (\ref{deltahypothesis}). Hence, necessarily there exists $e\in\mc E^+_{\mc U}$ such that $\tilde f^*_e\ge\lambda_0\alpha\beta_{\theta}^{1-n}$, and choosing $v_{j^*}$ to be the unique node in $\mc V$ such that $e\in\mc E^-_{v_{j^*}}$ one sees that (\ref{existseinDv}) holds true with $j=j^*$. Iterating this argument until $v_{j^*}=n$ proves the claim. 
\end{proof}\medskip
 
It is now easy to see that Lemma \ref{lemma:LBgammathetaalpha} implies that $\lim_{\alpha\downarrow0}\gamma_{\alpha,\theta}\ge C(\mc N)$ for every $\theta\ge1$, thus showing that $\gamma_0(f^{\circ},\mc G)\ge C(\mc N)$. Combined with Proposition \ref{propUB}, this shows that $\gamma_0(f^{\circ},\mc G)=C(\mc N)$, thus completing the proof of Theorem \ref{maintheo-weakstability}.

\section{Conclusion}
\label{sec:conclusion}
In this paper, we studied robustness properties of dynamical flow networks, where the dynamics on every link is driven by the difference between the inflow, which depends on the upstream routing decisions, and the outflow, which depends on the particle density, on that link.  We proposed a class of locally responsive distributed routing policies that rely only on local information about the network's current particle densities and yield the maximum weak resilience with respect to adversarial disturbances that reduce the flow functions of the links of the network. We also showed that the weak resilience of the network in that case is equal to min-cut capacity of the network, and that it is independent of the local information constraint and the initial flow. 
Strong resilience of dynamical flow networks is studied in the companion paper \cite{PartII}. 

 

 \bibliographystyle{ieeetr}%
  \bibliography{KS-transportation}

\end{document}